\documentclass[a4paper]{article}

\usepackage[ruled,vlined,linesnumbered]{algorithm2e}
\usepackage{multirow} 
\usepackage{lmodern}
\usepackage{tikz}
\usepackage{balance}
\usepackage{url}
\usepackage{amssymb}
\usepackage{times}
\usepackage{todonotes}
\usepackage{authblk}
\usepackage{amsthm}
\usepackage{epstopdf}

\usepackage{pgfplots}
\pgfplotsset{filter discard warning=false}

\pgfplotsset {
graph/.style={
height=57mm,
width=82mm,
legend style={draw=none, fill=none, font=\scriptsize},
  ticklabel style={font=\tiny},
xlabel near ticks,
ylabel near ticks},
}

\pgfplotsset {
smallgraph/.style={
height=48mm,
height=50mm,
width=61mm,
legend style={draw=none, fill=none, font=\scriptsize},
  ticklabel style={font=\tiny},
xlabel near ticks,
ylabel near ticks,
},
}

\newcommand{\ol}[1]{\overline{#1}}
\newcommand{\ul}[1]{\underline{#1}}

\newcommand{\ar}{\gamma}

\newcommand{\co}{\mbox{Cost}}

\newtheorem{theorem}{Theorem}
\newtheorem{lemma}[theorem]{Lemma}
\newtheorem{proposition}{Proposition}
\newtheorem{corollary}{Corollary}
\newtheorem{definition}{Definition}
\newtheorem{example}{Example}

\begin{document}

\title{Ordering Selection Operators Using the Minmax Regret Rule}



\author[1]{Khaled H. Alyoubi}
\author[2]{Sven Helmer}
\author[1]{Peter T. Wood}
\affil[1]{
Department of Computer Science and Information Systems,
Birkbeck, University of London,
London WC1E 7HX, UK}
\affil[2]{
Faculty of Computer Science,
Free University of Bozen-Bolzano,
39100 Bolzano, Italy}

\date{}

\maketitle

\begin{abstract}
Optimising queries in real-world situations under imperfect conditions is
still a problem that has not been fully solved. We consider
finding the optimal order in which to execute a given set of selection operators
under partial ignorance of their selectivities. The selectivities are modelled
as intervals rather than exact values and we apply a concept from decision
theory, the minimisation of the maximum regret, as a measure of optimality.
We show that the associated decision problem is NP-hard,
which renders a brute-force approach to solving it
impractical.  Nevertheless, by investigating properties of the problem
and identifying special cases which can be solved in
polynomial time, we gain insight that we use to
develop a novel heuristic for solving the general problem.
We also evaluate minmax regret query optimisation experimentally, showing that
it outperforms a currently employed strategy of optimisers that uses
mean values for uncertain parameters.

\end{abstract}


\section{Introduction}
\label{sec:intro}

Although query optimisation in database management systems (DBMSs) has been
a topic of research for decades, there are still important unresolved
issues. In his recent blog post \cite{Lohman14}, Guy Lohman 
highlights errors made in estimating cardinalities as a crucial 
factor. These kinds of errors cause optimisers to generate 
query execution plans that are way off the target in terms of efficiency. 
Consequently, an optimiser should try to avoid potentially bad
plans rather than strive for an optimal plan based on unreliable information.

For typical workloads, a DBMS can compile statistical data over time 
to obtain a fairly accurate picture. For instance,
estimating the selectivities of simple predicates on base relations
in a relational database is fairly well understood and can be done quite
accurately \cite{Garo02,Ioan03}. However, the situation changes once systems
are confronted with very unevenly distributed data values or 
predicates that are complex.

Trying to estimate selectivities in dynamic settings, such as data streams
\cite{Sriv05}, or in non-relational contexts, such as XML databases
\cite{Poly02,Zhang05}, also
poses challenges. It may even be impossible
to obtain any statistical data, because the query is
running on remote servers \cite{Zado02}. 
Detailed information may also not be available because a user issues an
atypical ad-hoc query or utilises parameter markers in a query. We propose to
use techniques from decision theory for making decisions under 
ignorance\footnote{Sometimes these are also called 
	decisions under uncertainty. We refer to them as decisions under ignorance to
	distinguish them from probability-based methods.}, 
meaning that we know what the alternatives and their outcomes
are, but we are unable to assign concrete probabilities to 
them \cite{Peterson09}.

In our approach we propose to build a robust 
query optimiser that is aware of the
unreliability of database statistics and considers
this during optimisation. 
When executing a query, the DBMS encounters a particular
instance of concrete parameter values: we call
this a {\em scenario}. The problem is that, during
the prior optimisation step, the optimiser does not know which scenario
the DBMS will face during plan execution.  Additionally, 
it is highly unlikely that there is a single 
execution plan that will yield the optimal cost 
for every potential scenario. Consequently,
our goal is to choose a query execution plan
that performs reasonably well regardless of the
scenario it encounters. More specifically,
we try to minimise the difference between
the cost of a plan $p$ and the cost of the optimal
plan when $p$ is executed under its worst-case
scenario. This is called {\em minmax regret
	optimisation} (MRO), which is a well-known technique for making decisions
under ignorance.  Previous work on query optimisation has considered
measures of robustness for query plans~\cite{Babcock05,Babu05,Markl04}, 
but not in terms of MRO.

In this paper, we focus on the selection operator $\sigma$,
an operator common to many data querying languages.    
Selection is sometimes called a filter operator
in contexts such as data stream processing \cite{AvHe00,BMMNW04} and
sensor networks \cite{DGHM05}, where there is renewed interest in  improving
the efficiency of processing these operators. A very common setting is
determining the order in which to apply a set of commutative filters to a
stream or a set of data items, e.g.\ tuples of a relation, so as to keep the
processing costs to a minimum.

There are well-known techniques for ordering selection operators to filter out
as many tuples as possible as early as possible at the lowest possible cost
\cite{HelSto93}. However, these techniques rely on having accurate values for the
operators' selectivities, i.e.,\ the percentage of tuples passing a filter, and
their processing costs (per tuple). Getting the estimation of selectivities
(and/or costs) wrong can lead to high overall costs for the pipelined
execution.

Our technique is based on using {\em intervals\/} rather than exact values for
describing selectivities, aiming at generating query plans that are minmax regret optimal.
However, identifying such plans, even for selection ordering, turns out to be
NP-hard.  As a result, we leave the investigation of further operators for future work
and focus first on finding a good heuristic for MRO selection ordering.

Intervals can provide a useful way to model selectivities when exact values
are unknown or hard to compute.  For example, Babu et al.~\cite{Babu05} compute
intervals from single-point estimates in order to model levels of uncertainty
regarding the accuracy of estimates, based on how such estimates were derived.
Moerkotte et al.~\cite{Moerkotte14} consider histograms which
guarantee a maximum multiplicative error (called the q-error)
for cardinality estimates.  Given such an estimate, the true cardinality
(selectivity) can easily be modelled by an interval, as we show in
Section~\ref{sec:related}.

For another situation in which interval selectivities arise, consider
estimating the selectivities of string predicates which perform substring matching
using SQL {\tt like}, a problem known to be difficult~\cite{CGG04}.  As an example,
let us consider a database in which
email messages are stored in a relation {\tt emails}, with attributes
such as {\tt sender}, {\tt subject} and {\tt body}
(the textual contents of the email).  Assume that many queries use selection
predicates such as {\tt subject like `\%invest\%'}, so the database maintains
indexes on words and on 2-grams (say) of words which allow it also to provide
selectivities for these.

Although the database maintains an index on words, the selectivity for the word
`invest' will be an underestimate for the selectivity of {\tt subject like `\%invest\%'}
since the strings `reinvest' and `investigation' (and many others) also match this
predicate.  Even if we are able to enumerate all words containing the string
`invest', we do not know how to combine their individual selectivities into
a single selectivity.  Instead we can use an interval selectivity with the
exact match as a lower estimate.  As the upper estimate, we can use the minimum
selectivity of all the 2-grams of `invest' since any string containing `invest' must
contain all of its 2-grams as well.

\begin{example}\label{ex:enron}
	As a concrete example, consider the following query on the
	Enron email data\footnote{\url{http://www.cs.cmu.edu/~./enron/}}:
	\begin{verbatim}
	select sender
	from   emails
	where  body like `%action%' and
	body like `%like%' and
	subject like `%use%';
	\end{verbatim}
	Let us denote the three predicates by $A$, $L$ and $U$ (for `action',
	`likes' and `use').  The interval selectivities for the three predicates,
	as computed using the method proposed above and explained in more detail
	in Section~\ref{sec:Max-min_Exp}, are
	$[0.03, 0.68]$ for $A$, $[0.17, 0.27]$ for $L$ and $[0.0008, 0.06]$ for $U$.
	Even if we consider only the upper and lower bounds of these intervals,
	they give rise to 8 possible scenarios.  No single plan (order) is optimal
	for all 8 scenarios, so the best we can do is find the plan which minimises
	the maximum regret.  This plan corresponds to the order $UAL$.  The maximum
	regret for this plan arises in the scenario when $U$ has its maximum
	selectivity, while $A$ and $L$ have their minimum selectivities (in this case,
	the predicates $U$ and $A$ should be swapped to get the optimal order).
	
	In the case of the above query, our heuristic finds the minmax regret optimal
	solution. 
	By way of contrast, an alternative heuristic such as that which takes the midpoints of the intervals and produces an optimal ordering based on those, produces the plan $ULA$.
	This plan has a maximum regret which is
	44\% worse than the minmax regret optimal plan.\hfill$\Diamond$
\end{example}

We should mention that the technique of using intervals can be applied to
other approximate or error-tolerant queries as well.  All we need is the
selectivity for an exact query as the lower bound and the selectivity
for a query that determines a candidate set with false positives as the upper bound.

Our contributions in this paper are as follows: 

\begin{itemize}
	
	\item We formalise the problem of optimal selection ordering
	under partial ignorance, i.e.,\ when selectivities are given
	as intervals.
	
	\item We identify a number of properties of the problem, including
	that (i)~only {\em extreme\/} scenarios (i.e., in which each operator
	takes on its minimum or maximum selectivity) need to be considered,
	(ii)~operators which {\em dominate\/} others (i.e., both their maximum
	and minimum selectivities are smaller) must appear before the
	dominated ones 
	in any optimal plan, and (iii)~the decision version of the problem
	is NP-hard.
	
	\item We investigate a number of special cases in which selection
	ordering under partial ignorance can be solved in polynomial time.
	Along the way, we also identify other important properties of scenarios
	in MRO selection ordering.
	
	\item Based on our findings we develop efficient optimisation heuristics,
	which we evaluate experimentally, using synthetic data, the Enron email
	data, and the Star Schema Benchmark (SSB) \cite{SSB13}.  
	The experiments demonstrate the benefit of
	using minmax regret optimisation, in some cases halving the 
	deviation from the optimal plan compared to conventional techniques.
	
\end{itemize}

The remainder of this paper is organised as follows.  We start by
reviewing related work on selection ordering and optimisation techniques
in the next section.  In Section~\ref{sec:formalDef}, we formalise the
problem of selection ordering under partial ignorance, using
minmax regret optimisation as the criterion for optimality.
Various properties of the problem are identified in Section~\ref{sec:properties}, while the proof of NP-hardness is given in Section~\ref{sec:hardness}.
Section~\ref{sec:specialCases} presents some special cases of the problem
which can be solved in polynomial time.  Our heuristic algorithm is given
in Section~\ref{sec:Max-min_Heuristic}, with its experimental evaluation
presented in Section~\ref{sec:Max-min_Exp}.  Finally, we conclude in
Section~\ref{sec:conclusion}.

\section{Background and Related Work}
\label{sec:related}

We assume we are given a set $S = \{\sigma_1, \sigma_2, \dots, \sigma_n\}$
of selection operators, or equivalently a conjunctive predicate
$p_1 \wedge p_2 \wedge \cdots p_n$.  The {\em selectivity\/} $s_i$ of
operator $\sigma_i$ or predicate $p_i$ is the fraction of tuples
that satisfy the operator or predicate.  Associated with each operator
$s_i$ is also a cost $c_i$, which is the cost per tuple of evaluating
the operator. 

Most database systems keep statistics allowing them to estimate
the selectivity for single attributes fairly accurately.  For the
joint selectivity of multiple attributes, much early work and
many systems make the {\em attribute value independence\/} (AVI)
assumption.  This assumes that the selectivity of a set of
operators $\{ \sigma_{i_1}, \sigma_{i_2} \ldots \sigma_{i_m} \}$ is
equal to $s_{i_1} \times s_{i_2} \times \cdots \times s_{i_m}$.
If instead a system stores (some) joint selectivities (it is
infeasible for it to store all of them), we can use the AVI
assumption to ``fill in the gaps'' or use the estimation approach
advocated in~\cite{Entropy07}.

\subsection{Selection Ordering}
\label{sec:SelectionOrdering}

Assuming we have accurate values for the selectivity $s_i$ and cost $c_i$
of selection operator $\sigma_i$, we can calculate
the {\em rank\/} $r_i$ of $\sigma_i$:
\begin{eqnarray}
	\label{eq:rank}
	r_i = (s_i - 1) / c_i
\end{eqnarray}
Given a set of selection operators, sorting and executing them
in non-decreasing order of their ranks results in the
minimal expected pipelined processing cost \cite{KBZ86} under the AVI assumption.
Clearly, the computation of the ranks and the sorting can be done in polynomial time.
A similar argument applies if a query uses a conjunction of predicates
on the same relation, and query evaluation uses a simple table scan.
In such a case, the optimiser should test the predicates in the order which
minimises the total number of tests.
Basically, ordering selection operators optimally is a solved problem, but
only when given exact values for the $s_i$ and $c_i$.

Similar optimisation problems have been studied in the context of sequential
testing. Here the goal is to find faulty components as quickly as possible by
testing them one by one. Each component has a probability of working correctly
and a cost for testing it. One of the earliest proposed solutions \cite{John56}
relies on ranking the components and then ordering them by their ranks, very
similar to the selection ordering described above.

\subsection{Optimising under Uncertainty}

In the following, we review different approaches for dealing with uncertain
parameters during query optimisation.
A common approach of many optimisers is to use the
mean or modal value of the parameters and then find the plan with least cost
under the assumption that this value remains constant during query execution,
an approach called Least Specific Cost (LSC) in \cite{p138-chu}. As Chu et
al.\ point out in \cite{p138-chu}, if the parameters vary significantly, 
this does not guarantee finding the plan of least expected cost.

An alternative
is to use probabilistic information about the parameters fed into the
database optimiser, an approach known as Least Expected Cost (LEC)
\cite{p138-chu}. (A discussion regarding the circumstances under which LEC or
LSC is best appears in \cite{p293-chu}.) 
In decision-theoretic terms, we are making decisions under risk, 
maximising the expected utility. However,
probability distributions for the possible parameter values are needed 
to make this approach work, whereas in our case we do not have these
prerequisites.

In parametric query optimisation
several plans can be precompiled and then, depending
on the query parameters, be selected for
execution \cite{Gang98}. 
However, if there is a large number of optimal
plans, each covering a small region of the 
parameter space, this becomes problematic.
First of all, we have to store all these plans.
In addition, constantly switching from one plan
to another in a dynamic environment (such as stream
processing) just because
we have small changes in the parameters introduces
a considerable overhead. In order to amend this,
researchers have proposed reducing the number
of plans at the cost of slightly decreasing the
quality of the query execution \cite{Harish07}.
Our approach can be seen as an extreme form of parametric query optimisation
by finding a single plan that covers the whole parameter space.

Another approach to deal with the lack of reliable statistics
is adaptive query processing, in which an execution plan
is re-optimised while it is running \cite{AvHe00,Babu05,Kabra98,Markl04}. 
It is far from trivial to determine at which point to 
re-optimise and adaptive query processing may also 
involve materialising large
intermediate results. More importantly, this means modifying the whole query
engine; in our approach no modifications of the actual query processing are
needed.  A gentler 
approach is the incremental execution of a query
plan \cite{Neu13}. Deciding on how to decompose
a plan into fragments and putting them together is 
still a complex task, though.

Estimates based on intervals arise explicitly in~\cite{Babu05}
and implicitly in~\cite{Moerkotte14}.  As mentioned in the
Introduction, Babu et al.~\cite{Babu05} use intervals to
model uncertainty in the accuracy of a single-point estimate.
Uncertainty is represented by a value from 0 (none) to 6 (very high).
Upper and lower bounds for the single-point estimate are then calculated
using the estimate and the uncertainty value.  During optimisation,
only three scenarios, those using the low estimates, the exact estimates and the high estimates,
are considered, rather than all scenarios as in our approach.
Moerkotte et al.~\cite{Moerkotte14} study histograms which provide
so-called q-error guarantees.  Given an estimate $\hat{s}$ for $s$,
the q-error of $\hat{s}$ is $\max(s/\hat{s}, \hat{s}/s)$.  An
estimate is $q$-acceptable if its q-error is at most $q$.  So if
an estimate $\hat{s}$ is $q$-acceptable, the true value $s$ lies in the
interval $1/q \times \hat{s} \leq s \leq q \times \hat{s}$.

Notions of robustness in query optimisation have been considered
in~\cite{Babcock05,Babu05,Markl04}.  Babcock and Chaudhuri~\cite{Babcock05}
use probability distributions derived from sampling as well as user preferences
in order to tune the predictability (or robustness) of query plans versus their performance. 
For Markl et al.~\cite{Markl04}, robustness means not continuing to execute to completion 
a query plan which is found to be suboptimal during evaluation; instead re-optimisation is performed.  
On the other hand, Babu et al.~\cite{Babu05} consider a plan to be robust only if its cost is within 
e.g.\ 20\% of the cost of the optimal plan.  None of these papers consider robustness in the sense of MRO. 
Moreover, these techniques need additional statistical information to work.


\subsection{Optimising under Ignorance}

Minmax regret optimisation (MRO) has been applied to a 
number of optimisation problems where some of
the parameters are (partially) unknown \cite{survey}.  
The complexity of the MRO version
of a problem is often higher than that of the original problem. 
Many optimisation problems with polynomial-time solutions turn out to be 
NP-hard in their MRO versions~\cite{survey}.

One example is minimising the {\em total flow time} (TFT),
in which $n$ jobs are scheduled on a single 
machine~\cite{Kasperski-book}.
The {\em flow time\/} of a job is the sum of its
processing time and the time it has had to wait before starting execution.
The total flow time is the sum of the flow times of all $n$ jobs.  
This scheduling problem can be solved in polynomial time given
exact job lengths (by sorting the jobs in non-decreasing
order of their processing times~\cite{averbakh}), but 
becomes NP-hard in its MRO variant~\cite{averbakh}.
Researchers have developed
approximation algorithms for the problem; for example, a 2-approximation
algorithm, bounding the approximate solution to be 
no more than twice the optimal solution,
is proposed in~\cite{Kasperski-book}.  

Among all MRO problems, TFT is the one closest to the problem we are
investigating. However, there are substantial differences: the formula for
computing the cost of a schedule is much simpler for TFT,
and the approach chosen to obtain a 2-approximation does not guarantee
a bound for MRO selection ordering, as we show in Section~\ref{sec:properties}.

\section{Selection Ordering MRO}
\label{sec:formalDef}

In this section we give a formal definition of the generalised selection
ordering problem with partially defined selectivities. The 
exact costs of selection
operators can also be unknown, but for the moment
we restrict ourselves to partially defined selectivities.


\subsection{Basic Definitions}
\label{sec:basic-defs}

We start out with definitions for selection operators with interval
selectivities and basic properties.

\begin{definition}
	Given a set $S = \{\sigma_1, \sigma_2, \dots, \sigma_n\}$ of selection
	operators, each has a selectivity $s_i$ and a cost $c_i$. Each selectivity is
	defined by a closed interval: for $1 \leq i \leq n$, $s_i = [\ul{s}_i,
	\ol{s}_i]$ with $\ul{s}_i, \ol{s}_i \in [0,1]$ and $\ul{s}_i \leq
	\ol{s}_i$. For $1 \leq i \leq n$, $c_i \in \mathbf{R}^+$
	represents the cost of $\sigma_i$ for processing an input tuple.
\end{definition}

Depending on their selectivity intervals selection operators may relate to
each other in a special way.  Later on we exploit this property in order
to optimise selection orders.

\begin{definition}
	Given two selection operators $\sigma_i, \sigma_j \in S$,
	we say that $\sigma_i$ {\em dominates\/} $\sigma_j$ if
	$\ul{s}_i \leq \ul{s}_j$ and $\ol{s}_i \leq \ol{s}_j$.
	The set $S$ of operators is called {\em dominant\/} if 
	for each pair $\sigma_i, \sigma_j \in S$ it is the case 
	that either $\sigma_i$ dominates $\sigma_j$ or 
	$\sigma_j$ dominates $\sigma_i$.
\end{definition}

Later on, it will be helpful to consider a special case of
dominant sets of operators.

\begin{definition}
	Given two selection operators $\sigma_i, \sigma_j \in S$,
	we say that $\sigma_i$ {\em strictly dominates\/} $\sigma_j$ if $\ol{s}_i \leq \ul{s}_j$.
	A {\em strictly dominant\/} set is defined analogously to a dominant set.
\end{definition}

If for two selection operators $\sigma_i, \sigma_j \in S$, neither 
$\sigma_i$ dominates $\sigma_j$ nor $\sigma_j$ 
dominates $\sigma_i$, then $\sigma_i$ and $\sigma_j$ form a {\em
	nested\/} pair of operators.  So, operator $\sigma_i$ is {\em nested\/} in
$\sigma_j$ if $\ul{s}_j < \ul{s}_i$ and $\ol{s}_i < \ol{s}_j$.

\begin{example}\label{ex:definitions}
	Let $S = \{\sigma_1, \sigma_2, \sigma_3\}$ be a set of selection operators,
	with selectivities $s_1=[.2, .8]$, $s_2=[.3, .5]$ and $s_3=[.1, .4]$.
	Operator $\sigma_3$ dominates both $\sigma_1$ and $\sigma_2$,
	but does not strictly dominate either of them.  Because
	$\sigma_2$ is nested in $\sigma_1$, the set $S$ is not dominant.
	\hfill$\Diamond$
\end{example}

\begin{definition}
	An assignment of a concrete value to each of the $n$ selectivities is called a {\em
		scenario} and is defined by a vector $x = (s_1, s_2, \dots, s_n)$, with
	$s_i \in [\ul{s}_i, \ol{s}_i]$. 
\end{definition}

Every time we actually run a query, we
encounter one scenario. However, during the optimisation step we are unaware
of which scenario we will face.
The set of all possible scenarios can be described by
$X = \{ x \mid x \in [\ul{s}_1, \ol{s}_1] \times[\ul{s}_2, \ol{s}_2] 
\times \dots
\times[\ul{s}_n, \ol{s}_n]\}$. There are certain scenarios we are particularly
interested in:

\begin{definition}
	A scenario $x_{ext} = (s_1, s_2, \dots, s_n)$ is called an {\em extreme scenario} 
	if, for each $1 \leq i \leq n$, $s_i$ is equal to either $\ul{s}_i$ or $\ol{s}_i$.
\end{definition}

Let $\pi^n$ be the set of all possible permutations over $1,2, \dots, n$. For
$\pi_j \in \pi^n$, $\pi_j(i)$ denotes the $i$-th element of $\pi_j$. 

\begin{definition}
	A query
	execution plan $p_j$ is a permutation $\sigma_{\pi_j(1)}$, $\sigma_{\pi_j(2)},
	\dots, \sigma_{\pi_j(n)}$ of the $n$ selection operators. The set of all
	possible query execution plans is given by
	\[
	P = \{ p \mid p = \sigma_{\pi(1)}, \sigma_{\pi(2)}, \dots, \sigma_{\pi(n)} \mbox{ such that } \pi \in \pi^n \}.
	\]
\end{definition}

\noindent
The cost of evaluating plan $p_j$ under a given scenario $x$ is
\begin{eqnarray}
	\label{eq:cost}
	\co(p_j,x) & = & \Omega (c_{\pi(1)} + s_{\pi(1)} c_{\pi(2)} +
	s_{\pi(1)} s_{\pi(2)} c_{\pi(3)} \nonumber \\ 
	&  & \: + \: \cdots \: + \prod_{i=1}^{n-1} s_{\pi(i)} c_{\pi(n)}) \nonumber \\
	& = & \Omega \left( \sum_{i=1}^n 
	\left( \prod_{j=1}^{i-1} s_{\pi(j)} \right) c_{\pi(i)} \right)
\end{eqnarray}
$\Omega$ is the cardinality of the relation on which we execute the selection
operators. Currently we make the AVI assumption that the selection predicates are
stochastically independent.  Extending our approach to situations in which
(some) joint selectivities are known is a topic for future work.


\begin{example}\label{ex:cost}
	Recall the set $S = \{\sigma_1, \sigma_2, \sigma_3\}$ of selection operators
	from Example~\ref{ex:definitions},
	with selectivities $s_1=[.2, .8]$, $s_2=[.3, .5]$ and $s_3=[.1, .4]$.
	There are 8 extreme scenarios for this example, one being given by
	scenario $x_1 = (\ul{s}_1, \ul{s}_2, \ul{s}_3) = (.2, .3, .1)$.
	One the the 6 possible plans for $S$ is given by plan 
	$p_1 = \sigma_1 \sigma_2 \sigma_3$.  Assuming that $\Omega$
	and each cost $c_i$ is set to $1$, we can calculate the
	cost of plan $p_1$ under scenario $x_1$, $\co(p_1,x_1)$,
	using Equation~(\ref{eq:cost}) as follows:
	\[
	\co(p_1,x_1) = (1 + .2  + .2 \times .3 ) = 1.26
	\]
	\hfill$\Diamond$
\end{example}

Let $p_{opt(x)}$ stand for the query execution plan having the
minimal cost for scenario $x$, and let $\pi_{opt(x)}$ be the permutation of the
selection operators for this plan. Since we are facing multiple scenarios,
the criterion for evaluating the optimality of a plan $p_j$ is
different to the one used in the classical selection ordering problem.
We utilise minmax regret optimisation to determine the quality
of a plan.

\subsection{Minmax Regret Optimisation}

Below we define the regret for a plan given a scenario, the
maximal regret for a plan, and finally the problem of finding
a plan that minimises the maximal regret.

\begin{definition}
	Given a plan $p$ and a scenario $x$, the absolute {\em regret\/} $\ar(p,x)$ of
	$p$ for $x$ is:
	\begin{eqnarray}
		\label{eq:regret}
		\ar(p,x) = \co(p,x) - \co(p_{opt(x)},x)
	\end{eqnarray}
	where $p_{opt(x)}$ is the optimal plan for scenario $x$.
	The maximal regret of a plan is the regret for its worst-case
	scenario and is simply defined as $\max_{x \in X} (\ar(p,x))$.
\end{definition}

\begin{definition}
	Given the set $P$ of all possible execution plans and the set $X$ of all
	possible scenarios, minimising the maximal regret is done as follows
	(where $R(P,X)$ is the optimal regret):
	\begin{eqnarray*}
		R(P,X) = & \min_{p \in P} (\max_{x \in X} (\ar(p,x)))
	\end{eqnarray*}
	Given a set $S$ of selection operators, let $P(S)$ denote the set of possible
	plans for $S$ and $X(S)$ denote the set of possible scenarios for $S$.  Then
	the {\em minmax regret optimisation\/} problem for $S$, which we denote
	$MRO(S)$, is to find a plan whose maximum regret matches $R(P(S),X(S))$.  For
	simplicity and when there is no confusion, we also use $MRO(S)$ to denote
	a plan which minimises $R(P(S),X(S))$.
\end{definition}

\begin{example}\label{ex:regret}
	Recall once again the set $S = \{\sigma_1, \sigma_2, \sigma_3\}$ of selection operators
	from Examples~\ref{ex:definitions} and~\ref{ex:cost},
	with selectivities $s_1=[.2, .8]$, $s_2=[.3, .5]$ and $s_3=[.1, .4]$.  For
	simplicity, assume that all operators have the same cost $1$ and that the relation
	has cardinality $\Omega = 1$ (so to get the real costs, the numbers in 
	Table~\ref{table:regretTable} 
	just have to be multiplied by the true cardinality).  
	To find the plan which minimises the maximum
	regret, we can perform an exhaustive enumeration of all possible execution
	plans under every possible scenario.  We show later in
	Theorem~\ref{th:extreme} that it is sufficient to consider only the extreme
	scenarios since the worst case scenario for any plan is always an extreme one.
	Hence, if there are $n$ operators, we need to consider $n!$ different execution
	plans under each of $2^n$ extreme scenarios.  For our example,
	Table~\ref{table:regretTable} shows the 48 regret values for the 6 possible
	plans under each of 8 extreme scenarios.

	For example, recall from Example~\ref{ex:cost} that the cost of
	the first plan $p_1 = \sigma_1 \sigma_2 \sigma_3$ under
	scenario $x_1 = (\ul{s}_1, \ul{s}_2, \ul{s}_3) = (.2, .3, .1)$ is $1.26$.
	The optimal plan $p_{opt(x)}$ for any scenario $x$ is one in which the
	operators are in non-decreasing order of their selectivities. Therefore, the
	optimal plan for scenario $x_1$ is $p_{opt(x_1)} = \sigma_3 \sigma_1 \sigma_2$
	and its cost is:
	\[
	\co(p_{opt(x_1)},x_1) = (1 + .1  + .1 \times .2 ) = 1.12
	\]
	The regret of plan $p_1$ under scenario $x_1$ using Equation~(\ref{eq:regret}) is:
	\begin{eqnarray*}
		\ar(p_1,x_1) & = & \co(p_1,x_1) - \co(p_{opt(x_1)},x_1) \\
		& = & 1.26 - 1.12 = 0.14
	\end{eqnarray*}
	In order to find the minmax regret solution, the maximum regret of each plan
	needs to be found. For plan $p_1$, the maximum regret is $1.05$ which occurs
	in scenario $(\ol{s}_1, \ol{s}_2, \ul{s}_3)$, its worst-case scenario.  The
	maximum regret for each plan is shown in bold face in
	Table~\ref{table:regretTable}.

	Finally, we are looking for the plan with the smallest maximum regret
	(i.e.,\ the smallest value in the last column of
	Table~\ref{table:regretTable}).  As a result the minmax regret solution,
	$MRO(S)$, is plan $\sigma_3 \sigma_1 \sigma_2$, which has the best performance
	among all plans when confronted with their worst-case scenarios.
	\hfill$\Diamond$
\end{example}

\begin{table}
	\centering
	\setlength{\tabcolsep}{2.2pt}
	\begin{tabular}{ | l | c | c | c | c | c | c | c | c || c | }
		\hline
		&$\ul{s}_1$ & $\ul{s}_1$ & $\ul{s}_1$ & $\ul{s}_1$ & $\ol{s}_1$ & $\ol{s}_1$ & $\ol{s}_1$ & $\ol{s}_1$ & Max \\
		& $\ul{s}_2$ & $\ul{s}_2$ & $\ol{s}_2$ & $\ol{s}_2$ & $\ul{s}_2$ & $\ul{s}_2$ & $\ol{s}_2$ & $\ol{s}_2$ & Regret\\
		& $\ul{s}_3$ & $\ol{s}_3$ & $\ul{s}_3$ & $\ol{s}_3$ & $\ul{s}_3$ & $\ol{s}_3$ & $\ul{s}_3$ & $\ol{s}_3$ &  \\
		\hline
		$\sigma_1\sigma_2\sigma_3$ & 0.14 & 0 & 0.18 & 0.02 & 0.91 & 0.62 & \textbf{1.05} & 0.6 & 1.05 \\ \hline
		$\sigma_1\sigma_3\sigma_2$ & 0.1 & 0.02 & 0.1 & 0 & \textbf{0.75} & 0.7 & 0.73 & 0.52 & 0.75 \\ \hline
		$\sigma_2\sigma_1\sigma_3$ & 0.24 & 0.1 & 0.48 & 0.32 & 0.41 & 0.12 & \textbf{0.75} & 0.3 & 0.75 \\ \hline
		$\sigma_2\sigma_3\sigma_1$ & 0.21 & 0.16 & \textbf{0.43} & 0.42 & 0.2 & 0 & 0.4 & 0.1 & 0.43 \\ \hline
		$\sigma_3\sigma_1\sigma_2$ & 0 & 0.22 & 0 & 0.2 & 0.05 & \textbf{0.3} & 0.03 & 0.12 & \textbf{0.3} \\ \hline
		$\sigma_3\sigma_2\sigma_1$ & 0.01 & 0.26 & 0.03 & \textbf{0.32} & 0 & 0.1 & 0 & 0 & 0.32 \\ \hline
	\end{tabular}
	\caption{The regret for each plan under each scenario in
		Example~\ref{ex:regret}.}
	\label{table:regretTable}
\end{table}


In the above example, it is interesting to consider which scenario gives rise
to the maximum regret for each plan.  Note that for each plan its worst-case
scenario is one in which the operators in some initial sequence in the plan
each take on their maximum selectivity followed by the remaining operators
taking on their minimum selectivity.  We call such a scenario a {\em
	max-min\/} scenario.

\begin{definition}
	Let $p$ be the plan $\sigma_{\pi(1)}, \sigma_{\pi(2)}, \dots,
	\sigma_{\pi(n)}$.  A scenario for $p$ is called a {\em max-min\/} scenario if
	there is a $0 \leq k \leq n$ such that for all $1 \leq i \leq k$, $s_{\pi(i)}
	= \ol{s}_{\pi(i)}$, and for all $k+1 \leq i \leq n$, $s_{\pi(i)} =
	\ul{s}_{\pi(i)}$.  
\end{definition}

So the first $k$ operators in $p$ take on their maximum
selectivity, while the rest take on the minimum.  
Note that for a plan $p$ with $n$ operators, there are $n+1$
max-min scenarios.  Max-min scenarios are the only scenarios considered by the
max-min heuristic we develop in this paper.  However, it is important to state
that, in general, the worst-case scenario for a plan may not be a max-min
scenario.


\section{Properties of MRO}
\label{sec:properties}

Before presenting algorithms for solving the MRO selection
ordering problem, we identify some of its important
properties.

\subsection{Extreme Scenarios}

In order to determine the worst-case scenario of a plan,
i.e., the scenario for which a plan exhibits its largest
regret, we only have to check extreme scenarios.

\begin{theorem}
	\label{th:extreme}
	The worst-case scenario for any query plan $p$ is always an extreme scenario.
\end{theorem}

\begin{proof}
	We introduce the following notation to show that our cost formulas are
	piecewise linear functions:
	\begin{eqnarray*}
		L_x^{\pi}(y) & := & \sum_{i=1,i\not=x}^{y} 
		\left( \prod_{j=1,j\not=x}^{i} s_{\pi(j)} \right) \\
		R_x^{\pi}(y) & := & \sum_{i=y,i\not=x}^{n-1} 
		\left( \prod_{j=1,j\not=x}^{i} s_{\pi(j)} \right) 
	\end{eqnarray*}
	where $L_x^{\pi}(y)$ computes the cost of plan $p$ with
	the operator permutation $\pi$ up to the operator
	at position $y$. We skip the operator at position
	$x$, i.e., the summand in which $s_{\pi(x)}$ appears
	first is left out of the sum and $s_{\pi(x)}$ 
	is omitted in all products. Analogously, we
	define $R_x^{\pi}(y)$ which computes the cost
	to the end of the plan starting from position $y$. If we do not
	want to skip any operators, we simply write
	$L^{\pi}(y)$ or $R^{\pi}(y)$.

	Expressing the costs of $p$ and $p_{opt(x)}$
	as a function of $s_m$:
	\begin{eqnarray*}
		\co(p,x,s_m) & = & L^{\pi}(v-1) + s_m R_v^{\pi}(v-1) \\
		\co(p_{opt(x)},x,s_m) & = & L^{\pi_{opt(x)}}(w-1) + s_m R_w^{\pi_{opt(x)}}(w-1) 
	\end{eqnarray*}
	we see that $\co(p,x)$ is a linear function in $s_m$. 
	$\co(p_{opt(x)},x,s_m)$ is linear as long as 
	$s_{\pi_{opt(x)}}(w-1) \leq s_m \leq s_{\pi_{opt(x)}}(w+1)$.
	If $s_m$ leaves this range, then $p_{opt(x)}$ will change,
	as all operators are sorted in ascending order of their
	selectivity. Nevertheless, $\co(p_{opt(x)},x,s_m)$ is a
	piecewise linear function. 
	Clearly, we can swap the positions of 
	two operators in an optimal plan without changing its
	optimality if the operators have exactly the same selectivity.
	So if $s_m = s_{\pi_{opt(x)}}(w-1) = \dots = s_{\pi_{opt(x)}}(w-k)$,
	then we can swap $\sigma_m$ with $\sigma_{w-k}$.
	Analogously, if $s_m = s_{\pi_{opt(x)}}(w+1) = \dots = s_{\pi_{opt(x)}}(w+k)$,
	then we can swap $\sigma_m$ with $\sigma_{w+k}$. For
	the cost of the optimal plan, this means
	$\co(p_{opt(x)},x,s_m)$
	\begin{eqnarray*}
		& = & \left \{
		\begin{array}{l@{\quad}l}
			\vdots & \\
			L^{\pi_{opt(x)}}(w-k-1) + s_m R_w^{\pi_{opt(x)}}(w-k-1) & \\
			\quad\quad\quad\mbox{if }
			s_{\pi_{opt(x)}}(w-k-1) \leq s_m \leq s_{\pi_{opt(x)}}(w-k) \\
			L^{\pi_{opt(x)}}(w-1) + s_m R_w^{\pi_{opt(x)}}(w-1) & \\
			\quad\quad\quad\mbox{if }
			s_{\pi_{opt(x)}}(w-1) \leq s_m \leq s_{\pi_{opt(x)}}(w+1) \\
			L_w^{\pi_{opt(x)}}(w+k) + s_m R_w^{\pi_{opt(x)}}(w+k) & \\
			\quad\quad\quad\mbox{if }
			s_{\pi_{opt(x)}}(w+k) \leq s_m \leq s_{\pi_{opt(x)}}(w+k+1) \\
			\vdots & \\
		\end{array}
		\right.
	\end{eqnarray*}
	
	\noindent
	Figure~\ref{fig:linear} illustrates the cost functions for 
	$p$ and $p_{opt(x)}$.
	
	\begin{figure}
		\begin{center}
			\mbox{\includegraphics[width=0.8\textwidth]{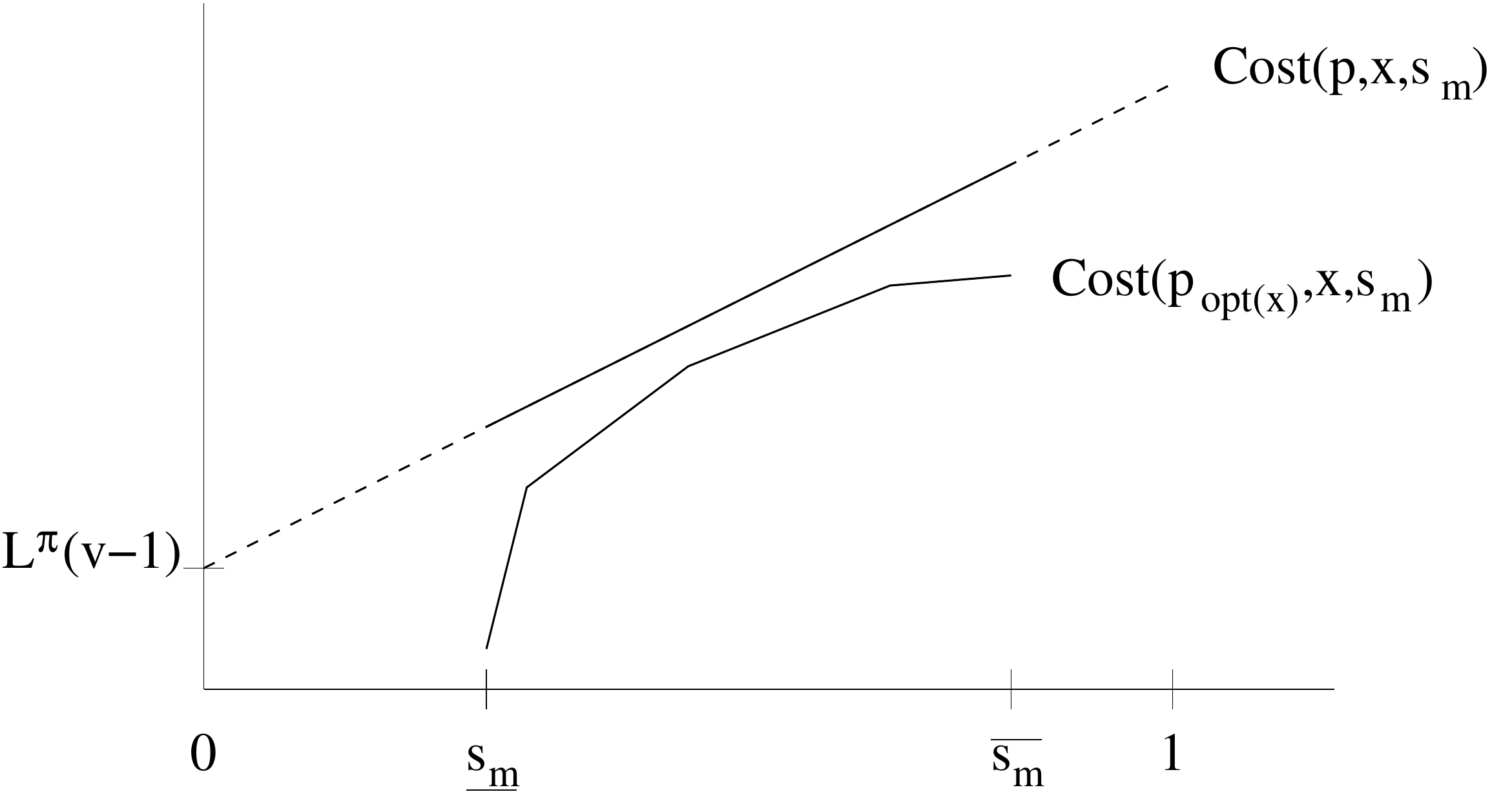}}
		\end{center}
		\caption{Visualisation of $\co(p,x,s_m)$ and $\co(p_{opt(x)},x,s_m)$}
		\label{fig:linear}
	\end{figure}
	
	We show that $\co(p_{opt(x)},x,s_m)$ is a concave (or convex
	upwards) function. For our piecewise linear function this
	means proving that by increasing $s_m$ (moving into the
	next piece) the slope will never increase, while if we
	decrease $s_m$, the slope will never decrease.

	Increasing $s_m$ will change the slope of
	$\co(p_{opt(x)},x,s_m)$ from $R_w^{\pi_{opt(x)}}(w-1)$ to
	$R_w^{\pi_{opt(x)}}(w+k)$. $R_w^{\pi_{opt(x)}}(w+k)$ is less or
	equal than $R_w^{\pi_{opt(x)}}(w-1)$, as they are identical
	except for the additional summands $w-1$ to $w+k-1$ in
	$R_w^{\pi_{opt(x)}}(w-1)$. Analogously, decreasing
	$s_m$ will change the slope from
	$R_w^{\pi_{opt(x)}}(w-1)$ to $R_w^{\pi_{opt(x)}}(w-k-1)$ (which
	is greater or equal).

	So $\ar(p,x) = (\co(p,x) - \co(p_{opt(x)},x))$ is a convex
	function, whose domain is restricted to a
	polyhedral convex set, defined by the lower and upper bounds of the
	selectivities. The global maximum of such a function is always found at one of
	the extreme points of the polyhedral convex set (Corollary 32.3.4 
	in~\cite{Rock70}).
\end{proof}

\subsection{Domination}

We can determine the relative order two operators have to be in
to minimise the maximal regret if one operator dominates the other.

\begin{theorem}
	\label{th:dominate}
	If $\sigma_a$ dominates $\sigma_b$, then there exists a plan
	$p$ minimising the maximal regret in which $\sigma_a$ precedes $\sigma_b$.
\end{theorem}

\begin{proof}
	Assume that $p$ is a plan minimising the maximal
	regret in which 
	$\sigma_b$ precedes $\sigma_a$: 
	$\pi(w)=b$ and $\pi(w+k)=a$.
	Furthermore, assume that $p'$ is constructed
	from $p$ by swapping $\sigma_b$ and $\sigma_a$: 
	$\pi'(w)=a$ and $\pi'(w+k)=b$. All the other
	operators are in exactly the same order as
	in $p$. We assume that $p'$ does {\em not\/} 
	minimise the maximal regret.

	Let us investigate the difference 
	in regret between $p'$ and $p$
	for any given scenario $x$.
	Since the optimal plan
	is the same for both regrets
	\begin{eqnarray*}
		\co(p',x) - \co(p,x) = (s_a - s_b) \sum_{i=w}^{w+k-1} 
		\left( \prod_{j=1, j\not=w}^{i} s_{\pi(j)} \right)
	\end{eqnarray*}
	\noindent
	we only need to check what happens between
	positions $w$ and $w+k-1$, as
	$L^{\pi}(w-1)$ and $R^{\pi}(w+k)$ (see the
	proof of Theorem~\ref{th:extreme} for the
	meaning of this notation) are identical
	for both $p$ and $p'$.

	Because this holds for every scenario, it also holds for
	the worst-case scenario $y'$ of $p'$. 
	Let us assume first that the selectivities for 
	$\sigma_a$ and $\sigma_b$ in $y'$ are either
	$(\ul{s_a},\ul{s_b})$, $(\ul{s_a},\ol{s_b})$,
	or $(\ol{s_a},\ol{s_b})$. Since $\sigma_a$
	dominates $\sigma_b$, we know that
	$\ul{s_a} \leq \ul{s_b}$, $\ul{s_a} \leq \ol{s_b}$,
	and $\ol{s_a} \leq \ol{s_b}$. Therefore,
	\[
	(s_a - s_b) \sum_{i=w}^{w+k-1} 
	\left( \prod_{j=1, j\not=w}^{i} s_{\pi(j)} \right) 
	\leq 0
	\]
	which means that the
	maximal regret of $p'$ cannot be greater than
	that of $p$. This also holds for strict domination,
	i.e., when $\ol{s_a} \leq \ul{s_b}$. However, this is a contradiction to
	our assumption. Thus, for the worst case scenario
	$y'$, we must have the selectivities $(\ol{s_a}, \ul{s_b})$
	with $\ol{s_a} > \ul{s_b}$.
	
	Let us look at a scenario $y''$ which is identical to 
	$y'$ except for $s_a = s_b = s_{c}$
	with $\ol{s_a} > s_{c} > \ul{s_b}$ (see Figure~\ref{fig:sasb}). From 
	Theorem~\ref{th:extreme} we know that the regret can
	be increased by moving to an extreme scenario.
	In this case $s_a$ has to be increased from $s_{c}$
	to $\ol{s_a}$ and $s_b$ has to be decreased from $s_{c}$
	to $\ul{s_b}$ to reach the maximal regret $\ar(p',y')$.
	
	\begin{figure}
		\begin{center}
			\mbox{\includegraphics[width=0.7\textwidth]{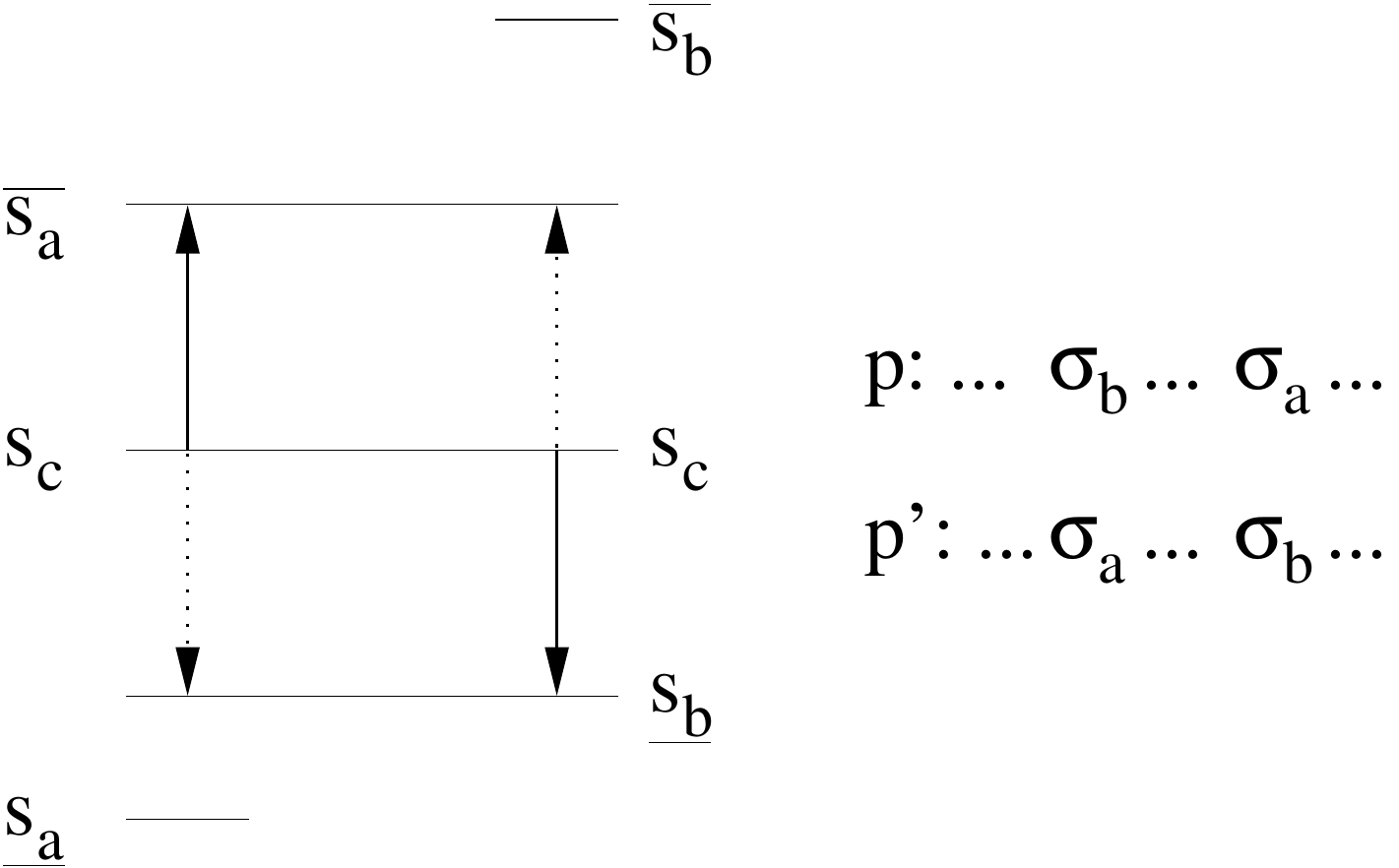}}
		\end{center}
		\caption{Visualisation for scenario $y''$}
		\label{fig:sasb}
	\end{figure}

	Clearly, $\ar(p,y'') = \ar(p',y'') $. 
	The following is illustrated in Figure~\ref{fig:sasb}.
	Increasing $s_b$
	from $s_{c}$ to $\ol{s_a}$ and decreasing $s_a$
	from $s_{c}$ to $\ul{s_b}$ for $p$ under scenario
	$y''$ (dotted arrows) will have exactly the same effect as
	increasing $s_a$ and decreasing $s_b$ for $p'$ under
	scenario $y''$ (solid arrows).
	However, this may not be an extreme
	case scenario for $p$ yet. Further increasing
	$s_b$ to $\ol{s_b}$ and decreasing $s_a$ to $\ul{s_a}$
	can never decrease the regret (according to 
	Theorem~\ref{th:extreme}). But that means we have
	found a scenario for $p$ which has at least the
	same regret as the worst-case scenario for
	$p'$, which contradicts our assumption.
\end{proof}

\begin{example}\label{ex:dominate}
	Recall from Example~\ref{ex:regret} the set $S = \{\sigma_1, \sigma_2, \sigma_3\}$ of selection operators,
	with selectivities $s_1=[.2, .8]$, $s_2=[.3, .5]$ and $s_3=[.1, .4]$.  Because $\sigma_3$ dominates
	$\sigma_1$ and $\sigma_2$, in the minmax regret solution, i.e.\ plan $\sigma_3 \sigma_1 \sigma_2$,
	$\sigma_3$ precedes $\sigma_1$ and $\sigma_2$.  As a result of domination, in this example we would 
	only have to consider two plans when searching for the minmax regret solution.
	\hfill$\Diamond$
\end{example}

\subsection{Midpoints of Intervals}

For the TFT problem, Kasperski used the simple heuristic of sorting jobs
in non-decreasing order according to the
midpoints of their intervals, yielding
a 2-approx\-imation~\cite{Kasperski-book}.
This approach does not guarantee a bound 
for MRO selection ordering; as shown below,
the quality of the solution can become
arbitrarily bad.

Given $2n+1$ operators, the first $n$ operators have
the selectivities $\ul{s_i} = 0$ and 
$\ol{s_i} = 1$ ($1 \leq i \leq n$),
while the next $n$ operators 
have the selectivities $\ul{s_i} = \ol{s_i} = 0.5 + \epsilon$
($n+1 \leq i \leq 2n$) for some small $\epsilon$.
The final operator has a constant selectivity
of $1$ to guarantee that it will always be 
in last position, meaning that its selectivity
will not impact any further steps.

The midpoint heuristic will order the operators
in exactly this way, from 1 to $2n+1$. Clearly,
the worst-case scenario for this plan is when
$s_i$ is set to $1$ for $1 \leq s_i \leq n$. 
In the optimal plan for this scenario, the
operators $\sigma_i$ with $n+1 \leq i \leq 2n$
will be executed first.

The regret of this plan is computed as follows:

\begin{tabular}{ccccccccccc}
	& $1$ & \hspace*{-.3cm} $+$ \hspace*{-.3cm} & 
	$1^2$ & \dots & \hspace*{-.3cm} $+$ \hspace*{-.3cm} & 
	$1^n$ & \hspace*{-.3cm} $+$ \hspace*{-.3cm} & 
	$f(n)$ \\
	$-$ &
	$(0.5 + \epsilon)$ & \hspace*{-.3cm} $-$ \hspace*{-.3cm} 
	& $(0.5 + \epsilon)^2$ & 
	\dots & \hspace*{-.3cm} $-$ \hspace*{-.3cm} 
	& $(0.5 + \epsilon)^n$ & \hspace*{-.3cm} $-$ \hspace*{-.3cm}  
	& $g(n)$ 
\end{tabular}

\noindent
where $f(n)$ and $g(n)$ stand for the cost of the remaining
operators in the plan. A lower bound for this expression is
the following, since $f(n) \geq g(n)$ (see 
Lemma~\ref{lem:factor} below):
\begin{eqnarray*}
	n - n (0.5 + \epsilon)
\end{eqnarray*}
With increasing $n$ and small values for $\epsilon$, 
this expression can get arbitrarily large.

\begin{lemma}
	\label{lem:factor}
	Given a query plan $p$ and a scenario $x$, 
	we have the following relationship between the
	summands in $\co(p,x)$ and $\co(p_{opt(x)},x)$, where
	$p_{opt(x)}$ is the optimal plan for scenario $x$:
	\begin{eqnarray*}
		\prod_{j=1}^{k} s_{\pi(j)} \geq 
		\prod_{j=1}^{k} s_{\pi_{opt(x)}(j)} \mbox{ for all } k \mbox{ with }
		1 \leq k \leq n-1 
	\end{eqnarray*}
\end{lemma}

\begin{proof}
	If there exists an $s_m = 0$, then $\prod_{j=1}^{k} s_{\pi_{opt(x)}(j)} = 0$
	for all $k$ and the above holds, as $s_i \geq 0$ for all $i$.
	This is due to $s_{\pi_{opt(x)}(1)}=s_m=0$ 
	($p_{opt(x)}$ sorts the selections
	in non-decreasing order of their selectivities according
	to the ranking algorithm). 
	Thus, in the following all $s_i > 0$.

	We assume there is a $k$ for
	which $\prod_{j=1}^{k} s_{\pi(j)} < \prod_{j=1}^{k} s_{\pi_{opt(x)}(j)}$
	(proof by contradiction).
	Let $\pi^k = \{\pi(i) \mid 1 \leq i \leq k \}$ be the set of indexes
	of the first $k$ selection operators in $p$ and
	$\pi_{opt(x)}^k = \{\pi_{opt(x)}(i) \mid 1 \leq i \leq k \}$ the set
	of indexes of the first $k$ selection operators in $p_{opt(x)}$.
	If $\pi^k = \pi_{opt(x)}^k$, then the two products are
	equal, which is a contradiction to our assumption.
	So in the following we assume $\pi^k \not= \pi_{opt(x)}^k$.

	Nevertheless, the intersection between
	$\pi^k$ and $\pi_{opt(x)}^k$ may be non-empty. In this
	case we can discard all the selectivities common
	to both products:
	\begin{eqnarray*}
		\prod_{j \in \pi^k \cap \pi_{opt(x)}^k} s_j 
		\prod_{j \in \pi^k \setminus \pi_{opt(x)}^k} s_j 
		& < &
		\prod_{j \in \pi^k \cap \pi_{opt(x)}^k} s_j
		\prod_{j \in \pi_{opt(x)}^k \setminus \pi^k} s_j
		\\
		\Leftrightarrow \;\;\; 
		\prod_{j \in \pi^k \setminus \pi_{opt(x)}^k} s_j
		& < &
		\prod_{j \in \pi_{opt(x)}^k \setminus \pi^k} s_j
	\end{eqnarray*}
	\noindent
	Therefore, there exists $i \in \pi^k \setminus \pi_{opt(x)}^k$
	such that $s_i < \max( s_l \mid l \in \pi_{opt(x)}^k \setminus \pi^k)$
	(or the inequality would not hold). 
	However, that means $\sigma_i$ has appeared in
	$p$ but not yet in $p_{opt(x)}$. This is a contradiction:
	the selection operators are sorted in non-decreasing order
	in $p_{opt(x)}$ and $\sigma_i$ should have appeared in
	$p_{opt(x)}$ before the selection operator with selectivity
	$\max( s_l \mid l \in \pi_{opt(x)}^k \setminus \pi^k)$.
\end{proof}

\section{Hardness of MRO}
\label{sec:hardness}

\newcommand{\regret}{{\sc minmax regret}}
\newcommand{\cover}{{\sc set cover}}
\newcommand{\rcover}{{\sc restricted set cover}}
\newcommand{\covert}{{\sc exact cover by 3-sets}}

In this section, we show that the decision problem for general $MRO(S)$, 
which we call \regret, is NP-hard.  In this version, we are 
given a set $S = \{\sigma_1, \sigma_2, \dots, \sigma_n\}$
of selection operators, with each operator $\sigma_i$ assumed
to have unit cost.  We are also given a set $X = \{X_1, X_2, \ldots, X_m \}$
of scenarios, where each scenario $X_j$ specifies a selectivity $s_{ij}$
for each operator $\sigma_i$, $1 \leq j \leq m$ and $1 \leq i \leq n$.

To simplify the notation, let us identify a plan $p$
with the permutation $\pi$ it defines, and from now on
use $\pi(i)$ to denote the index of the operator appearing
in position $i$ in plan $\pi$.

Below we define the decision problem \regret\ as well as
the well-known NP-complete problems \cover\ and \covert.

\vspace{1ex}\noindent
\regret:
given a set $S$ of $n$ selection operators, a set $X$ of $m$ scenarios,
and a real number $R$, is there a plan whose maximum regret is less than $R$?

\vspace{1ex}\noindent
\cover:
given a finite set $A$, a collection $T$ of subsets of $A$, and a positive integer $r$,
is there a subset $C = \{ C_1, \ldots, C_r \}$ of $T$ such that
$\bigcup_{C_i \in C} C_i = A$, that is, such that $C$ covers $A$?

\vspace{1ex}\noindent
\covert:
given a finite set $A$ with $|A|=3q$ and a collection $T$ of 3-element subsets of $A$,
is there a subset $C$ of $T$ such that each element of $A$ occurs in exactly one
member of $C$?

\vspace{1ex}
It is known that a restriction of \covert\
which requires that each element of the set $A$ appears in exactly
three subsets of $T$ is NP-complete~\cite{Gonzalez1985}.
Since \cover\ is a generalisation of \covert, we also have
that \rcover, defined below, is NP-complete.

\vspace{1ex}\noindent
\rcover:
given finite set $A$, collection $T$ of subsets of $A$
such that each element of $A$ appears in exactly three subsets
of $T$, and positive integer $r$,
is there a subset $C = \{ C_1, \ldots, C_r \}$ of $T$ such that
$C$ covers $A$?

\vspace{1ex}
We show that \regret\ is NP-hard by reducing \rcover\ to it.

\begin{theorem}
	\regret\ is NP-hard.
\end{theorem}

\begin{proof}
	We reduce \rcover\ to \regret.
	Given an instance of \rcover\ represented by $A$,
	$T$ and $r$, we construct an instance of \regret\ as follows.
	Let $|A|=m$ and $|T|=n$.
	Each subset $C_j$ in $T$ is represented by an operator $\sigma_j$ in $S$,
	and each element $a_i \in A$ is represented by a scenario $X_i \in X$
	such that the selectivity for operator $\sigma_j$ in $X_i$, that is,
	$s_{ij}$ is $1/(n+1)$ if $a_i \in C_j$ and $1$ if $a_i \not\in C_j$.
	Since each element of $A$ appears in exactly three subsets, each
	scenario $X_i \in X$ has three selectivities of $1/(n+1)$ and $n-3$
	selectivities of $1$.  Hence the optimal plan for each scenario has
	the same cost, say, $p$.  We set $R$ to $r - p$ and
	claim that there is a subset of $T$ of size $r$ which covers $A$
	if and only if there is a plan whose maximum regret over all scenarios
	is less than $R$.
	
	Assume there is a subset $C = \{ C_{k_1}, \ldots, C_{k_r} \}$ of $T$
	which covers $A$.  Let $\pi$ be any plan in which $\pi(i) = \sigma_{k_i}$,
	$1 \leq i \leq r$, that is, in which the first $r$ operators correspond
	to subsets in the cover.  Since $C$ is a cover, for no scenario $X_i$
	can it be the case that the selectivity for each of the first
	$r$ operators in $\pi$ is $1$.  At worst, the first $r-1$ operators
	have selectivity $1$, with the $r$'th operator having selectivity
	$1/(n+1)$, and the remaining $n-r$ operators having selectivity $1$.
	The cost of this plan is therefore $r-1 + (n-r+1) / (n+1)$, which
	is always less than $r$.  Hence the regret is less than $R = r-p$,
	where $p$ is the cost of the optimal plan.
	
	Now assume no subset of $T$ of size $r$ covers $A$.  In other words,
	for every subset of size $r$, at least one element of $A$ is not in
	any set in the subset.  Hence, for every plan $\pi$, there must be
	some scenario in which the first $r$ operators have selectivity $1$.
	Finding a plan which minimises the maximum regret is the same as
	finding a plan which minimises the maximum cost since the cost of
	the optimal plan is the same for each scenario.
	Since every plan in this case has cost at least $r$, there is no plan
	whose maximum cost is less than $r$.  Hence there is no plan whose
	maximum regret is less than $R = r-p$, where $p$ is the cost of the optimal plan.
\end{proof}

\section{Some Polynomial-Time Cases} 
\label{sec:specialCases}

In this section we show that, for sets of selection operators $S$
satisfying certain properties, $MRO(S)$ can be found in polynomial time.
In particular, we look at dominant operators, which can easily be ordered
correctly, and their combination with constant operators, i.e., operators for
which we can obtain exact selectivity values.
As before, we assume that the cost of each operator is one.


\subsection{Constant and Dominant Operators}
\label{sec:Constant_and_Dominant_Opers}
Let $S$ be a set of selection operators such that the selectivity of each operator can be estimated accurately (i.e., each selectivity is constant).  Then, as mentioned in Section~\ref{sec:SelectionOrdering}, $MRO(S)$ can be found by sorting the operators in non-decreasing order of their rank given by Equation~(\ref{eq:rank}). Given our assumption that each operator has cost one, finding $MRO(S)$ reduces to sorting the operators in non-decreasing order of their selectivity alone. 

Recall from Section~\ref{sec:basic-defs} the definition of a dominant set $S$ of operators.
Given a dominant set $S$ of operators, it follows from Theorem~\ref{th:dominate} that the 
minmax regret solution is one where the operators are sorted in non-decreasing order according 
to their minimum (or maximum) selectivity value.  (Note that a set of constant operators is a 
special case of a dominant set of operators.)  We therefore have:

\begin{corollary}
	If $S$ is a dominant set of $n$ operators, then $MRO(S)$ can be solved in $O(n \log n)$ time.
\end{corollary}

\subsection{Strictly Dominant Operators with a Constant Operator}
\label{sec:Strictly_Dominant_and_Constants}

When we include nested operators (recall the
definition from Section~\ref{sec:basic-defs}),
the problem becomes much more difficult.  As a step
in the direction of solving the general problem, we consider below the
simple case of a strictly dominant set of operators (also defined in
Section~\ref{sec:basic-defs}) along with
a single constant operator nested within one of the non-constant operators.
If $S$ is a strictly dominant set of operators, then the plan $MRO(S)$ has zero regret under all scenarios.
This is because all operators in $MRO(S)$ will be in the same position as in the
corresponding optimal plan under all scenarios.

Let $S$ be a strictly dominant set which includes a constant operator $\sigma_c$
nested within one of the non-constant operators, say $\sigma_i$.  In this case,
we know how to place the dominant operators relative to each other in $MRO(S)$
but we need to determine the position of $\sigma_c$ in $MRO(S)$.
Since $\ul{s}_i \leq s_c \leq \ol{s}_i$, the constant operator $\sigma_c$
should be either immediately before or immediately after $\sigma_i$ in $MRO(S)$.  Interestingly, the correct position for $\sigma_c$ depends only
on the midpoint of the selectivity $s_i$ of $\sigma_i$.

\begin{proposition}\label{prop:strict-dom-constant}
	Let $S$ be a strictly dominant set of $n$ operators such that
	$MRO(S) = (\sigma_1,\ldots,\sigma_n)$.  Let $\sigma_c$ be an operator
	with constant selectivity $s_c$ such that $\ul{s}_i \leq s_c \leq \ol{s}_i$,
	for some $1 \leq i \leq n$, and $S' = S \cup \{ s_c \}$.  In $MRO(S')$,
	$\sigma_c$ is placed between
	(1)~$\sigma_{i-1}$ and $\sigma_i$ if $s_c \leq (\ul{s}_i + \ol{s}_i)/2$, or
	(2)~$\sigma_i$ and $\sigma_{i+1}$ if $s_c \geq (\ul{s}_i + \ol{s}_i)/2$.
	
\end{proposition}

Note that if $s_c = (\ul{s}_i + \ol{s}_i)/2$, then $\sigma_c$ can be placed either between $\sigma_{i-1}$ and $\sigma_i$ or between $\sigma_i$ and $\sigma_{i+1}$ in $MRO(S')$.

\begin{proof}
	The operators $\sigma_i$ and $\sigma_c$ are always neighbours
	in an MRO solution (appearing after $\sigma_{i-1}$ and before $\sigma_{i+1}$): 
	any operator $\sigma_{h}$ such that $1 \leq h \leq
	i-1$ dominates $\sigma_c$ and $\sigma_i$, and
	$\sigma_c$ and $\sigma_i$
	dominate any operator $\sigma_{j}$ such that $i+1 \leq j \leq
	n$.

	Let us assume that $\sigma_i$ has selectivity $\ul{s}_i$
	in scenario $x$ and selectivity $\ol{s}_i$ in 
	scenario $x'$. 
	Moreover, let us define plan
	$p$ that is constructed from $MRO(S)$ by placing $\sigma_c$ before
	$\sigma_i$, $\pi(v) = c$ and $\pi(v+1) = i$. Similarly we define plan $p'$ 
	by placing $\sigma_c$ after $\sigma_i$ (i.e. $\pi'(v) = i$ and $\pi'(v+1)
	= c$).

	{\em Case} 1: The optimal plan $p_{opt(x)}$ places $\sigma_i$
	before $\sigma_c$ in scenario $x$ (the operators are sorted in
	non-decreasing order of their selectivities), i.e.,
	$\sigma_i$ and $\sigma_c$ are at position $v$ and $v+1$, respectively.
	We now compute the maximum regret of $p$ and $p'$ for
	scenario $x$:
	\begin{small}
		\begin{eqnarray*}
			\co(p_{opt(x)},x) & = & 
			L^\pi(v-1) (1 + \ul{s}_i + \ul{s}_i s_c) + R^\pi(v+2) \\
			\co(p,x) & = & 
			L^\pi(v-1) (1 + s_c + s_c \ul{s}_i) + R^\pi(v+2) \\
			\co(p',x) & = & 
			L^{\pi'}(v-1) (1 + \ul{s}_i + \ul{s}_i s_c) + R^{\pi'}(v+2) \\
		\end{eqnarray*}
	\end{small}
	As $L^\pi(v-1) = L^{\pi'}(v-1)$ and $R^\pi(v+2) = R^{\pi'}(v+2)$,
	we can compute the regret of $p$ and $p'$ as follows:
	\begin{eqnarray}
		\ar(p,x) & = & L^\pi(v-1) \left(s_c - \ul{s}_i\right)\label{eq:maxForScenario_x} \\
		\ar(p',x) & = & 0
	\end{eqnarray}
	So for scenario $x$, plan $p$ has a greater regret than $p'$ and it
	can be calculated by Equation~(\ref{eq:maxForScenario_x}).

	{\em Case} 2: In scenario $x'$, $\sigma_i$ follows $\sigma_c$ in
	$p_{opt(x')}$. For computing the costs of the different plans, this
	means:
	\begin{small}
		\begin{eqnarray*}
			\co(p_{opt(x')},x') & = & 
			L^\pi(v-1) (1 + s_c + s_c \ol{s}_i) + R^\pi(v+2) \\
			\co(p,x') & = & 
			L^\pi(v-1) (1 + s_c + s_c \ol{s}_i) + R^\pi(v+2) \\
			\co(p',x') & = & 
			L^{\pi'}(v-1) (1 + \ol{s}_i + \ol{s}_i s_c) + R^{\pi'}(v+2) \\
		\end{eqnarray*}
	\end{small}
	Consequently, the regret of $p$ and $p'$ is
	\begin{eqnarray}
		\ar(p,x') & = & 0 \\
		\ar(p',x') & = & L^\pi(v-1)
		\left(\ol{s}_i - s_c\right) \label{eq:maxForScenario_y}
	\end{eqnarray}
	In this case (scenario $x'$) plan $p'$ has a greater regret
	than $p$ and it is calculated using
	Equation~(\ref{eq:maxForScenario_y}).

	Comparing both cases we can see that $p$ has a smaller
	maximum regret than $p'$ whenever 
	Eq.~(\ref{eq:maxForScenario_x}) $<$ Eq.~(\ref{eq:maxForScenario_y}).
	This is the case when $s_c < (\ul{s}_i + \ol{s}_i)/2$
	and then we place $\sigma_c$ before $\sigma_i$.
	Similarly, plan $p'$ has a smaller 
	maximum regret than $p'$ whenever 
	Eq.~(\ref{eq:maxForScenario_x}) $>$ Eq.~(\ref{eq:maxForScenario_y}).
	We place $\sigma_c$ after $\sigma_i$ when
	$s_c > (\ul{s}_i + \ol{s}_i)/2$.
	For the breakeven point, i.e. 
	$s_c = (\ul{s}_i + \ol{s}_i)/2$, 
	$\sigma_c$ can be placed before or after $\sigma_i$.
	
\end{proof}

Proposition~\ref{prop:strict-dom-constant} can be generalised to the case in
which each non-constant operator has at most one constant operator nested
within it. An interesting observation about the situation 
described in Proposition~\ref{prop:strict-dom-constant} is that
the worst-case scenario is a max-min scenario.

\begin{proposition}\label{prop:strict-dom-constant-scenario}
	Let $S$ be a strictly dominant set of $n$ operators such that
	$MRO(S) = (\sigma_1,\ldots,\sigma_n)$.  Let $\sigma_c$ be an operator with constant selectivity $s_c$ such that
	$\ul{s}_i \leq s_c \leq \ol{s}_i$, for some $1 \leq i \leq n$, and $S' = S
	\cup \{ s_c \}$.  
	The scenario $(\ol{s}_1,\ldots,\ol{s}_{j-1}, s_c,
	\ul{s}_{j},\ldots,\ul{s}_n)$,
	in which either $\sigma_{j-1}$ or $\sigma_j$ is equal
	to $\sigma_i$,
	is a worst-case scenario for $MRO(S')$. 
\end{proposition}

\begin{proof}
	From Proposition~\ref{prop:strict-dom-constant} we know that 
	if $s_c \leq (\ul{s}_i + \ol{s}_i)/2$, then the minmax regret is computed by
	Equation~\ref{eq:maxForScenario_x} for plan $p$. In plan $p$,
	$\sigma_i$ follows $\sigma_c$ (so $\sigma_j = \sigma_i$)
	and the selectivity of 
	$\sigma_i$ is $\ul{s}_i$. The other selectivities
	do not influence the regret, so we can set the selectivity 
	of the operators $\sigma_1$ to $\sigma_{j-1}$ to the upper 
	bound and the selectivity of the operators 
	$\sigma_{j+1}$ to $\sigma_{n}$ to the lower
	bound.
	If $s_c \geq (\ul{s}_i + \ol{s}_i)/2$, then the minmax regret is computed by
	Equation~\ref{eq:maxForScenario_y} for plan $p'$. In plan $p'$,
	$\sigma_c$ follows $\sigma_i$ (so $\sigma_{j-1} = \sigma_i$)
	and the selectivity of 
	$\sigma_i$ is $\ol{s}_i$. Here we choose the upper bounds
	for the operators $\sigma_1$ to $\sigma_{j-2}$ and
	the lower bounds for the operators
	$\sigma_{j}$ to $\sigma_{n}$. In both cases this
	results in a max-min scenario.
\end{proof}


\section{Max-min Heuristic} 
\label{sec:Max-min_Heuristic}

Computing the regret of every selection ordering for every possible scenario
makes the brute-force algorithm infeasible, since there
are $n!$ different orderings and $2^n$ scenarios, given $n$ operators. So in
order to find an efficient heuristic, we have to significantly reduce
the number of orderings and scenarios. While doing so, we want to leverage the
insights gained from our theoretical investigation.


Let us first look at the number of possible scenarios. As we have seen in the
previous section, max-min scenarios seem to play a special role when it comes
to the maximum regret of a given plan $p$. Intuitively this makes sense, as in
an optimal plan many of the operators $\sigma_i$ located towards the beginning of
$p$ with selectivities $\ol{s_i}$ will trade places with operators $\sigma_j$
located towards the end of $p$ with selectivities $\ul{s_j}$. Consequently,
there tends to be a large difference between the plan $p$ and an optimal plan
for a max-min scenario, leading to a substantial (if not maximal) regret for $p$.
So in our heuristic we aim to generate plans that perform well for
max-min scenarios. This reduces the number of scenarios we have to consider
from $2^n$ to $n+1$.

We now turn to determining the order of the selection operators. 
There are two well-known basic methods for doing this (efficiently).
The first one is constructing a plan by combining partial plans in a way that
leads to an optimised execution order. Very often putting the partial plans
together requires using a heuristic to solve a combinatorial problem. The
second method is to quickly create a complete plan (e.g., by using a simple
heuristic) and then try to improve the plan by rewriting it (e.g., by
swapping or removing and re-inserting operators). In our approach we wanted to
have both options available, so we decided to develop different variants. The
complexity of our heuristic shows slight differences depending on the variant
we use; however, the algorithms we apply all have polynomial complexity.

Our max-min heuristic algorithm, $H(p,q)$, which is in fact a template for a
number of algorithms, is shown as Algorithm~\ref{heuristic-alg}. It is
parameterised by two inputs: $p$, a (possibly empty) starting plan, and $q$, an
order in which to process operators. Clearly, to generate a complete plan
the union of $p$ and $q$ has to contain all the operators. If the intersection
of $p$ and $q$ is empty, our algorithm is similar to insertion
sort: in turn, we consider each operator in $q$ and place it into $p$ at the
position that minimises the regret over all max-min scenarios. 
If an operator in $q$ is already present in $p$, then we remove it from $p$
before re-inserting it. This is equivalent to moving an operator to a
different position. Again we determine the position minimising the regret over
all max-min scenarios.


\begin{algorithm}\caption{$H(p, q)$ \label{heuristic-alg}}
	\ForEach{operator $t$ from the sequence $q$}{
		\lIf{$t$ is in $p$}{remove $t$ from $p$}\;
		Assume $p$ currently comprises $i$ operators\;
		\ForEach{position $j$, $1 \leq j \leq i+1$, in $p$}{
			Temporarily insert $t$ in position $j$ in $p$\;
			\ForEach{max-min scenario for $p$}{
				Calculate the regret of plan $p$\;
				Store the maximum regret for position $j$\;
			}
		}
		Choose as the final position for $t$ in $p$ that which minimises the maximum regret\;
	}
	Return $p$\;
\end{algorithm}

It is clear that the max-min heuristic runs in polynomial-time.
For each partial plan comprising $i$ operators, we consider $i+1$
possible positions for the next operator.  In each of these positions,
we consider $i+2$ max-min scenarios.  Calculating the regret of a plan
with $n$ operators can be done in time $O(n \log n)$.  
Hence the algorithm described above has an overall complexity
of  $O(n^4 \log n)$
(in the worst case $i = n$ for every execution of the outer loop).
However, by computing costs incrementally when an operator moves position and
one max-min scenario moves to the next, we can implement the heuristic to run
in time $O(n^3)$.

\begin{example}\label{ex:heuristic}
	Recall from Example~\ref{ex:dominate} the set $S = \{\sigma_1, \sigma_2, \sigma_3\}$
	of selection operators, with selectivities $s_1=[.2, .8]$, $s_2=[.3, .5]$ and $s_3=[.1, .4]$. 
	Consider our max-min heuristic algorithm, $H(p,q)$, with initial plan $p = \sigma_3 \sigma_1$ 
	and remaining operator $q = \sigma_2$.
	Since $p$ consists of two operators, $\sigma_2$ should be checked in three positions:
	before $\sigma_3$, after $\sigma_1$ and between them.  For each position and resulting plan, 
	the regret is calculated under all max-min scenarios, of which there are four in this example.
	
	As an example, consider the plan in which $\sigma_2$ is placed between $\sigma_3$ and $\sigma_1$. 
	The regret will be calculated for the scenarios 
	$(\ul{s}_3, \ul{s}_2, \ul{s}_1)$, $(\ol{s}_3, \ul{s}_2, \ul{s}_1)$, 
	$(\ol{s}_3, \ol{s}_2, \ul{s}_1)$ and $(\ol{s}_3, \ol{s}_2, \ol{s}_1)$. 
	The maximum regret for this plan is $0.3$ which occurs in scenario $(\ol{s}_3, \ol{s}_2, \ul{s}_1)$. 
	
	Finally, the solution will be the plan with the smallest maximum regret,
	which happens to be $\sigma_3 \sigma_2 \sigma_1$. As a matter of fact, 
	the solution returned by the max-min heuristic is the same as the actual minmax regret solution,
	as was shown in Example~\ref{ex:regret}.	\hfill$\Diamond$
\end{example}

In the following two subsections, we consider various criteria
for choosing an initial plan and for ordering the remaining operators.

\subsection{Choosing an Initial Plan}
\label{sec:Max-min_BaseCase}

Even though we can run our heuristic with an empty initial plan $p$,
i.e., building a solution by inserting all operators one by one, often it makes
sense to start with a prebuilt partial plan.

One particular and important case is that of dominant operators. Given a set $S$ of
operators, if we can identify a subset $S' \subseteq S$ of dominant operators,
we know that we can find an optimal solution $p'$ for $S'$ quickly and that
the relative order of the operators in $p'$ will not change in any optimal
plan for $S$ (see Theorem~\ref{th:dominate}). Thus, taking $p'$ as the initial
plan when calling $H(p,q)$ makes good sense. However, there may be
different ways to choose $S'$, as in general there may be more than one such
dominant set. If we have more than one option, we can use the following criteria to
make a decision: choose the subset $S'$ (1) with the maximum cardinality or
(2) whose operators have the largest total width. As we often encountered
several subsets sharing the same maximum cardinality, we introduced a
tie-breaker: choose the subset $S'$ (3) with the maximum cardinality whose
total width is greatest.
In our experiments, we found that this third approach gave the best overall
results.

\begin{example}\label{ex:initialPlan}
	Recall from Example~\ref{ex:heuristic} the set $S = \{\sigma_1, \sigma_2, \sigma_3\}$
	of selection operators, with selectivities $s_1=[.2, .8]$, $s_2=[.3, .5]$ and $s_3=[.1, .4]$.
	Set $S$ has two dominant subsets: $S_1= \{\sigma_1, \sigma_3\}$ and $S_2= \{\sigma_2, \sigma_3\}$.
	Both obviously satisfy criterion~(1) above, being of maximum cardinality.
	However, if we use criterion~(2), namely the set which has operators with 
	the largest total selectivity width, then we will choose $S_1$ since its
	total width is $0.9$ while that of $S_2$ is $0.5$. $S_1$ would also be 
	chosen according to criterion~(3).
	
	After choosing the preferable subset, we need to produce initial plan $p$
	by sorting the operators in nondecreasing order of their minimum (or maximum)
	selectivities.  Therefore, $p = \sigma_3 \sigma_1$ when $S_1$ is chosen, 
	while $p = \sigma_3 \sigma_2$ if $S_1$ is chosen.  \hfill$\Diamond$
\end{example}

Having an initial plan allows us to combine our algorithm with other
heuristics. We can take the output of another algorithm as our initial plan
$p$ and then refine this result by running $H(p,q)$ on it. Moreover, we can
use the output of $H(p,q)$ as input for another iteration of our own
heuristic.

\subsection{Ordering Criteria}
\label{sec:Max-min_Sorting}

Since our algorithm makes only a single pass over all the operators when
(re-)inserting them into the plan, the order in which operators are considered
may have a significant impact on the final outcome. For example, when
inserting selections into an empty initial plan, operators considered earlier
are tested in fewer positions relative to each other compared to those considered
later.

We have considered two different ordering criteria in
our experiments: interval {\em midpoint\/} (denoted by $M$)
and interval {\em width\/} (denoted by $W$).
Given a selectivity interval $s = [\ul{s},\ol{s}]$,
the midpoint of $s$ is $(\ul{s} + \ol{s})/2$ while
the width of $s$ is $\ol{s} - \ul{s}$.  In each case,
operators can be ordered by non-decreasing (denoted $+$)
or non-increasing (denoted $-$) values.  Overall, the
ordering criteria are denoted by $M+$, $M-$, $W+$ and
$W-$.  So, for example, $W+$ stands for operators being
considered in non-decreasing order of their selectivity
interval width.  

\section{Experimental Results}
\label{sec:Max-min_Exp}

We evaluated the max-min heuristic experimentally, measuring the impact of
different parameters on its performance. We also implemented the brute-force
algorithm for finding optimal solutions in order to evaluate how well the
heuristic performs.

A commodity PC, with 8 GB RAM, Intel Core i5 processor running at
3.19 GHz and Windows 7 Enterprise (64-bit), was used to perform the experiments.
The minmax regret brute-force algorithm and max-min heuristic were 
implemented in Java and compiled with the Eclipse IDE (Juno release),
which is JDK compliant and uses the JavaSE-1.7 execution environment.
The SSB queries were run on a simulation platform written in Ruby 1.9.3.

\subsection{Generating Test Data}

We first generated a synthetic data set 
to investigate the performance of our heuristic.
Each test case corresponded to a set of $k$ selection operators, with $k$
ranging from 2 to 10, and for each $k$ we generated a hundred different
sets.  While $k=2$ is not hard to solve, it was included
for verification purposes (any heuristic has to be able to find the optimal
plan for this simple case). Ten operators was the upper limit we were able to
solve optimally, checking $10! \cdot 2^{10}$ ($\approx$ 3.7 billion) different
costs for each test case. 
For each set of selection operators we determined the lower and upper bounds
of their selectivity intervals by generating $2k$ uniformly distributed random
numbers between 0 and 1.


For real-world data, we used the Enron email data set, as introduced
in Example~\ref{ex:enron}.  Once again, test queries used from 2 to 10
operators/predicates.  For each $n \in [2,10]$, 20 queries
were generated, each with one predicate on {\tt subject} and $n-1$ predicates
on {\tt body}.  The 20 queries were generated by randomly
selecting from 40 keywords for {\tt subject} and 45 keywords
for {\tt body}, and were checked to ensure that each returned
a nonempty answer.

We also evaluated minmax regret optimisation using a version of the
Star Schema Benchmark (SSB) with data skew \cite{SSB13} (SSB itself is 
a variation of the TPC-H benchmark). We generated benchmark data with a
scaling factor of 1, meaning that the central facts table, {\it lineorder},
contains 6,000,197 tuples, and joined all dimensional tables to the 
{\it lineorder} table. We then randomly 
picked from two to ten attributes from a subset of all
available attributes to generate queries. Queries basically consist of a
conjunctive predicate whose clauses are made up of the selected attributes
compared to a random value taken from the attribute's domain, using a less-than
or greater-than operator. The following predicate is an example generated in
our experiments:\\
{\tt orderKey $<$ 2964443 and linenumber $>$ 5 and quantity $<$ 29}.

\subsection{Parameters}

For the synthetic and Enron data sets, we 
looked at the effects of the ordering criteria and the choice of 
initial plan on the quality of our heuristic. Additionally, we 
investigated the impact of running our heuristic multiple times, using the
output of one phase as the initial plan of the next phase.

We measure the performance of our heuristic by defining the {\em regret ratio}
$\lambda(S)$, which is the regret computed by $H(p,q)$ divided by the optimal
regret. More formally, given a set $S$ of selection operators, let us denote the
set of possible plans by $P(S)$ and the set of possible scenarios by $X(S)$.
Recall from Section~\ref{sec:basic-defs} that $R(P(S),X(S))$ then denotes the
optimal regret. Then
\begin{eqnarray*}
	\lambda(S) = \frac{R(H(p,q),X(S))}{R(P(S),X(S))}
\end{eqnarray*}
We only calculate $\lambda(S)$ using the above formula when the optimal minmax regret is non-zero.
As mentioned in Section~\ref{sec:specialCases}, the optimal minmax regret
is zero only when $S$ forms a strictly dominating set.  For such cases,
our max-min heuristic always finds the optimal minmax regret solution,
so we define $\lambda(S)$ to be one.

In view of having multiple test cases per number of selection operators, we
calculate the {\em average regret ratio} and the {\em worst regret ratio}
(simply the maximum value of $\lambda(S)$).

For the Enron data, we calculated selectivity intervals for the {\tt like}
predicates as described in the Introduction.  For each selected keyword,
we ran queries to find the minimum selectivity (given by exact matches
of the keyword) and the maximum selectivity (given by the minimum
selectivity of all 2-grams of the keyword).  This gave rise to a range of
intervals: those with small values such as $[0.0004, 0.01]$ for keyword
`progress' in the {\tt subject}, those with larger values
such as $[0.6, 0.7]$ for `you' in the {\tt body}, and those
with a big range such as $[0.07, 0.6]$ for `price' in the {\tt body}.


For the Star Schema Benchmark we created some very rudimentary histograms
by dividing the domain of an attribute into equal-sized ranges,
counting the number of tuples that fall into each range. We do not keep any
further information on the distribution of tuples within each range of a
histogram. For example, Figure~\ref{fig:ordtotalprice} shows the histogram
for the attribute ordtotalprice, consisting of 20 ranges each covering
roughly 18,000 different values, e.g., bucket \#1 covers the range from 1 to
17,673.

\begin{figure}
	\begin{center}
		\includegraphics[width=0.6\textwidth]{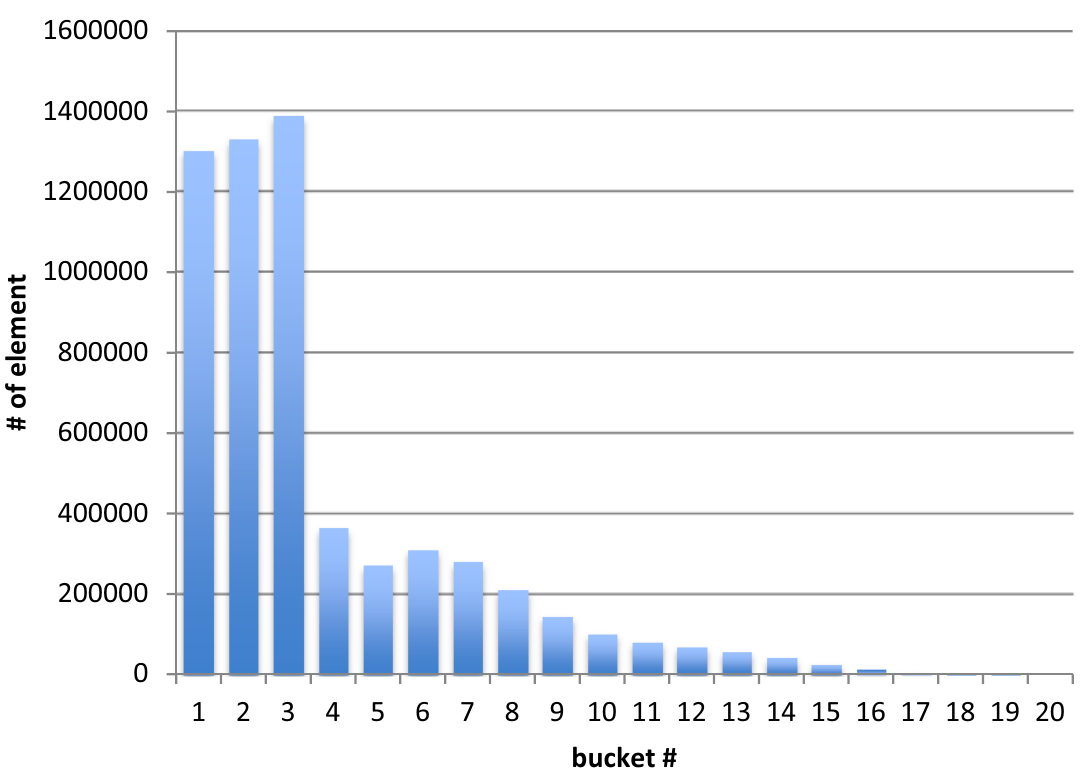}
	\end{center}
	\caption{Histogram for attribute {\tt ordtotalprice}.}
	\label{fig:ordtotalprice}
\end{figure}

This basic information allows us to determine intervals for the
selectivities of selection operators. For a 
``less than'' / ``greater than''
operator, we know that all histogram ranges exclusively 
covering smaller/larger values have to be included fully.
However, for the range the predicate value falls into, we do not know
precisely how many elements will be selected. In extreme cases,
none or all of the elements satisfy the predicate,
giving us the lower and upper bound for the selectivity.
Example~\ref{ex:histogram} illustrates this with concrete values.

\begin{example}\label{ex:histogram}
	Given the histogram for attribute $X$ below 
	and the predicate $X < 126$, we can compute the lower bound 
	and upper bound for the selectivity as follows:
	lower bound = $\frac{200}{1000} = 0.2$,
	upper bound = $\frac{200+100}{1000} = 0.3$.
	
	\begin{center}
		\begin{tabular}{|c|c|}
			\hline
			Range & \# of elements \\
			\hline
			\multicolumn{1}{|r|}{1-100} & 200 \\
			101-200 & 100 \\
			201-300 & 400 \\
			301-400 & 300 \\
			\hline
		\end{tabular}
	\end{center}
	\hfill$\Diamond$
\end{example}


Many sophisticated query optimisation techniques, such as least expected cost
(LEC), assume that they have access to probability distributions of parameter
values. LEC needs this to be able to compute utilities \cite{p138-chu}.
However, in our case we only have very rudimentary statistics, since we do not
know anything about the distribution of attribute values within a range. The
best we can do is to fall back on the assumption of uniform distribution,
approximating the distribution using a mean value (this is also what least
specific cost (LSC) optimisation would do in this case). For example, applying
this method to the numbers given in Example~\ref{ex:histogram} would yield a
selectivity of 0.225 for the predicate $X < 126$. We compare our minmax regret
optimisation technique to a mean-value-based approach using SSB data.
Additionally, we do a comparison with a simple midpoint heuristic, i.e.,
sorting the intervals in non-decreasing order of their midpoint.

\subsection{Results}

First we present the results obtained studying the different variants of the
max-min heuristic on the synthetic and Enron data, and then move on to the Star Schema Benchmark results.

\subsubsection{Synthetic and Enron Data Sets}


We experimented with a number of operator 
ordering criteria and initial plans for the max-min heuristic.
These included starting with an empty initial plan ($\emptyset$), considering random
operator ordering (U), ordering by midpoint (M- and M+) and ordering by width (W- and W+).  
We briefly summarise the findings of our experiments here. Overall,
the W+ ordering (non-decreasing width) performed best with an overall average
regret ratio of 1.03 and an overall worst regret ratio of 1.94. 
W- was often even worse than a random order,
while M+ and M- sometimes generated plans whose regret ratio was above 3. 
We also ran a midpoint heuristic that simply ordered the intervals in
non-decreasing order of their midpoints (not going through all max-min
scenarios). The midpoint heuristic was often worse than running the max-min
heuristic with a random order.

While W+ ordering performs better than the M+, \mbox{M-,}
and W- max-min heuristics and
the midpoint heuristic, it is still not significantly
better than the random ordering. In a second phase of our 
evaluation we seeded our heuristic with an initial plan. 
The results for initial plan D:CW with operator
ordering W+ were best
(D:CW stands for the largest subset of dominant operators,
and, in case of a tie, the one with the greatest total width of the
operators) in terms of the percentage of exact solutions 
and the average regret ratio. The results
for the worst case regret ratio were rather
inconclusive, so we tried to improve on this by running multiple phases of our
heuristic.

\begin{figure*}
	\begin{tabular}{ccc}

		\begin{tikzpicture}
		\begin{axis}[%
		smallgraph,
		xmode=normal,
		ymode=normal,
		ymin=-0.02,
		ymax=1.02,
		legend style={cells={anchor=west}, legend pos=south west, yshift=-1mm},
		legend plot pos=left,
		xlabel= total number of operators,
		ylabel= \% of cases with exact solution]
		\pgfplotstableread{./exact_sol.txt}\exactsol
		\addplot[color=blue,mark=diamond] table[x=x_axis, y=midp] {\exactsol};
		\addlegendentry{{\smaller[2]midpoint}}
		\addplot[color=darkgray,mark=o] table[x=x_axis, y=empty_U] {\exactsol};
		\addlegendentry{{\smaller[2]($\emptyset$,U)}}
		\addplot[color=violet,mark=star] table[x=x_axis, y=DW] {\exactsol};
		\addlegendentry{{\smaller[2](D:CW,W+)}}
		\addplot[color=orange,mark=square] table[x=x_axis, y=DW_W] {\exactsol};
		\addlegendentry{{\smaller[2]((D:CW,W+),W+)}}
		\addplot[color=black,mark=triangle] table[x=x_axis, y=DW_W_W] {\exactsol};
		\addlegendentry{{\smaller[2](((D:CW,W+),W+),W+)}}
		\end{axis}
		\end{tikzpicture}
		& \hspace*{-5.5mm} &
		\begin{tikzpicture}
		\begin{axis}[%
		smallgraph,
		xmode=normal,
		ymode=normal,
		ymin=0.9,
		ymax=3.0,
		legend style={cells={anchor=east}, legend pos=north west, yshift=-1mm},
		legend plot pos=right,
		xlabel= total number of operators,
		ylabel= worst regret ratio]
		\pgfplotstableread{./worst_regret.txt}\worstregret
		\addplot[color=blue,mark=diamond] table[x=x_axis, y=midp] {\worstregret};
		\addlegendentry{{\smaller[2]midpoint}}
		\addplot[color=darkgray,mark=o] table[x=x_axis, y=empty_U] {\worstregret};
		\addlegendentry{{\smaller[2]($\emptyset$,U)}}
		\addplot[color=violet,mark=star] table[x=x_axis, y=DW] {\worstregret};
		\addlegendentry{{\smaller[2](D:CW,W+)}}
		\addplot[color=orange,mark=square] table[x=x_axis, y=DW_W] {\worstregret};
		\addlegendentry{{\smaller[2]((D:CW,W+),W+)}}
		\addplot[color=black,mark=triangle] table[x=x_axis, y=DW_W_W] {\worstregret};
		\addlegendentry{{\smaller[2](((D:CW,W+),W+),W+)}}
		\end{axis}
		\end{tikzpicture}
		\\
		(a) Percentage of exact solutions &\hspace*{-5.5mm}&
		(b) Worst regret ratio \\
		\begin{tikzpicture}
		\begin{axis}[%
		smallgraph,
		xmode=normal,
		ymode=normal,
		ymin=0.99,
		ymax=1.2,
		legend style={cells={anchor=west}, legend pos=north west, yshift=-1mm},
		legend plot pos=left,
		xlabel= total number of operators,
		ylabel= average regret ratio]
		\pgfplotstableread{./avg_regret.txt}\avgregret
		\addplot[color=blue,mark=diamond] table[x=x_axis, y=midp] {\avgregret};
		\addlegendentry{{\smaller[2]midpoint}}
		\addplot[color=darkgray,mark=o] table[x=x_axis, y=empty_U] {\avgregret};
		\addlegendentry{{\smaller[2]($\emptyset$,U)}}
		\addplot[color=violet,mark=star] table[x=x_axis, y=DW] {\avgregret};
		\addlegendentry{{\smaller[2](D:CW,W+)}}
		\addplot[color=orange,mark=square] table[x=x_axis, y=DW_W] {\avgregret};
		\addlegendentry{{\smaller[2]((D:CW,W+),W+)}}
		\addplot[color=black,mark=triangle] table[x=x_axis, y=DW_W_W] {\avgregret};
		\addlegendentry{{\smaller[2](((D:CW,W+),W+),W+)}}
		\end{axis}
		\end{tikzpicture}
		& \hspace*{-5.5mm} &
		\begin{tikzpicture}
		\begin{axis}[%
		smallgraph,
		xmode=normal,
		ymode=normal,
		legend style={cells={anchor=east}, legend pos=north west, yshift=-1mm},
		legend plot pos=right,
		xlabel= total number of operators,
		ylabel= run time (seconds)]
		\pgfplotstableread{./run_time.txt}\avgregret
		%
		\addplot[color=blue,mark=diamond] table[x=x_axis, y=midpoint] {\avgregret};
		\addlegendentry{{\smaller[2]midpoint}}
		\addplot[color=darkgray,mark=o] table[x=x_axis, y=empty_U] {\avgregret};
		\addlegendentry{{\smaller[2]($\emptyset$,U)}}
		\addplot[color=violet,mark=star] table[x=x_axis, y=DW] {\avgregret};
		\addlegendentry{{\smaller[2](D:CW,W+)}}
		\addplot[color=orange,mark=square] table[x=x_axis, y=DW_W] {\avgregret};
		\addlegendentry{{\smaller[2]((D:CW,W+),W+)}}
		\addplot[color=black,mark=triangle] table[x=x_axis, y=DW_W_W] {\avgregret};
		\addlegendentry{{\smaller[2](((D:CW,W+),W+),W+)}}
		\end{axis}
		\end{tikzpicture}
		\\
		
		(c) Average regret ratio 
		&&
		(d) Run time 
		\\
	\end{tabular}
	\caption{Results for synthetic data set.}
	\vspace*{-.4cm}
	\label{fig:artificial}
\end{figure*}
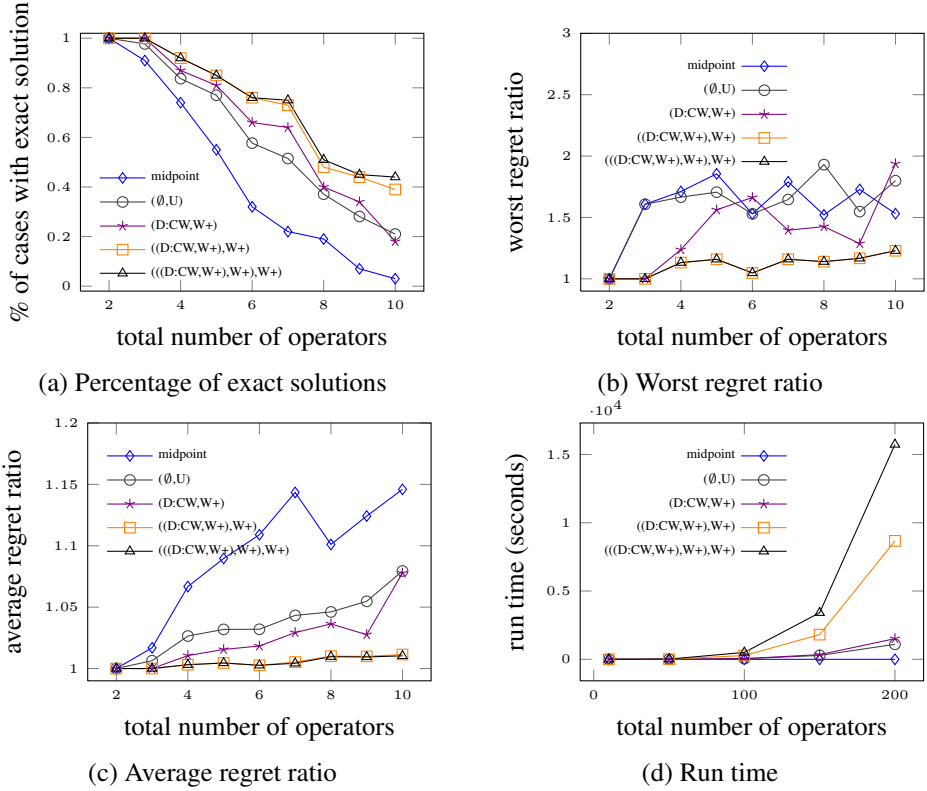


Figures~\ref{fig:artificial}(a), (b) and (c)
show the results for running our heuristic multiple times. This means that we
take the output of running one phase of our heuristic
and use it as the initial plan for the next phase. The figures show the
results for starting off by running
(D:CW, W+) first and then executing two more phases.

As can be seen, this variant clearly outperforms the baseline algorithm 
($\emptyset$,U), the midpoint heuristic, and
the other variants in all respects.   For example, for 10 operators, the
worst regret ratio is less than $1.23$ and the average ratio is approximately
$1.01$, compared to approximately $1.94$ and $1.08$, respectively, for running
only a single phase of the heuristic.  Moreover, running one additional phase
improves the quality of the generated plan significantly, but running another
phase 
makes almost no difference.

Figure~\ref{fig:artificial}(d) 
shows the run time of the W+ ordering variant (single
and multiple phases) together with the baseline algorithm ($\emptyset$,U) when
generating plans for up to 200 operators. Unsurprisingly, the variants
midpoint, ($\emptyset$,U), and (D:CW,W+) have the 
fastest run times, as they only sort a set of operators or execute
a single operator insertion phase. Furthermore, it can be clearly seen that
the additional run time of (((D:CW,W+),W+),W+) does not pay off, since it
produces plans that are only marginally better than those of 
((D:CW,W+),W+).


\begin{figure*}
	\begin{tabular}{ccc}

		\begin{tikzpicture}
		\begin{axis}[%
		smallgraph,
		xmode=normal,
		ymode=normal,
		ymin=-0.02,
		ymax=1.02,
		legend style={cells={anchor=west}, legend pos=south west, yshift=-1mm},
		legend plot pos=left,
		xlabel= total number of operators,
		ylabel= \% of cases with exact solution]
		\pgfplotstableread{./Enron_exact_sol.txt}\exactsol
		\addplot[color=blue,mark=diamond] table[x=x_axis, y=midp] {\exactsol};
		\addlegendentry{{\smaller[2]midpoint}}
		\addplot[color=darkgray,mark=o] table[x=x_axis, y=empty_U] {\exactsol};
		\addlegendentry{{\smaller[2]($\emptyset$,U)}}
		\addplot[color=violet,mark=star] table[x=x_axis, y=DW] {\exactsol};
		\addlegendentry{{\smaller[2](D:CW,W+)}}
		\addplot[color=orange,mark=square] table[x=x_axis, y=DW_W] {\exactsol};
		\addlegendentry{{\smaller[2]((D:CW,W+),W+)}}
		\addplot[color=black,mark=triangle] table[x=x_axis, y=DW_W_W] {\exactsol};
		\addlegendentry{{\smaller[2](((D:CW,W+),W+),W+)}}
		\end{axis}
		\end{tikzpicture}
		& \hspace*{-5.5mm} &
		\begin{tikzpicture}
		\begin{axis}[%
		smallgraph,
		xmode=normal,
		ymode=normal,
		ymin=0.9,
		ymax=3.0,
		legend style={cells={anchor=east}, legend pos=north west, yshift=-1mm},
		legend plot pos=right,
		xlabel= total number of operators,
		ylabel= worst regret ratio]
		\pgfplotstableread{./Enron_worst_regret.txt}\worstregret
		\addplot[color=blue,mark=diamond] table[x=x_axis, y=midp] {\worstregret};
		\addlegendentry{{\smaller[2]midpoint}}
		\addplot[color=darkgray,mark=o] table[x=x_axis, y=empty_U] {\worstregret};
		\addlegendentry{{\smaller[2]($\emptyset$,U)}}
		\addplot[color=violet,mark=star] table[x=x_axis, y=DW] {\worstregret};
		\addlegendentry{{\smaller[2](D:CW,W+)}}
		\addplot[color=orange,mark=square] table[x=x_axis, y=DW_W] {\worstregret};
		\addlegendentry{{\smaller[2]((D:CW,W+),W+)}}
		\addplot[color=black,mark=triangle] table[x=x_axis, y=DW_W_W] {\worstregret};
		\addlegendentry{{\smaller[2](((D:CW,W+),W+),W+)}}
		\end{axis}
		\end{tikzpicture}
		\\
		(a) Percentage of exact solutions &\hspace*{-5.5mm}&
		(b) Worst regret ratio \\
		\multicolumn{3}{c}{
			\begin{tikzpicture}
			\begin{axis}[%
			smallgraph,
			xmode=normal,
			ymode=normal,
			ymin=0.99,
			ymax=1.2,
			legend style={cells={anchor=west}, legend pos=north west, yshift=-1mm},
			legend plot pos=left,
			xlabel= total number of operators,
			ylabel= average regret ratio]
			\pgfplotstableread{./Enron_avg_regret.txt}\avgregret
			\addplot[color=blue,mark=diamond] table[x=x_axis, y=midp] {\avgregret};
			\addlegendentry{{\smaller[2]midpoint}}
			\addplot[color=darkgray,mark=o] table[x=x_axis, y=empty_U] {\avgregret};
			\addlegendentry{{\smaller[2]($\emptyset$,U)}}
			\addplot[color=violet,mark=star] table[x=x_axis, y=DW] {\avgregret};
			\addlegendentry{{\smaller[2](D:CW,W+)}}
			\addplot[color=orange,mark=square] table[x=x_axis, y=DW_W] {\avgregret};
			\addlegendentry{{\smaller[2]((D:CW,W+),W+)}}
			\addplot[color=black,mark=triangle] table[x=x_axis, y=DW_W_W] {\avgregret};
			\addlegendentry{{\smaller[2](((D:CW,W+),W+),W+)}}
			\end{axis}
			\end{tikzpicture}}
		\\
		\multicolumn{3}{c}{
			(c) Average regret ratio} \\
	\end{tabular}
	\caption{Results for Enron data set.}
	\vspace*{-.5cm}
	\label{fig:enron}
\end{figure*}
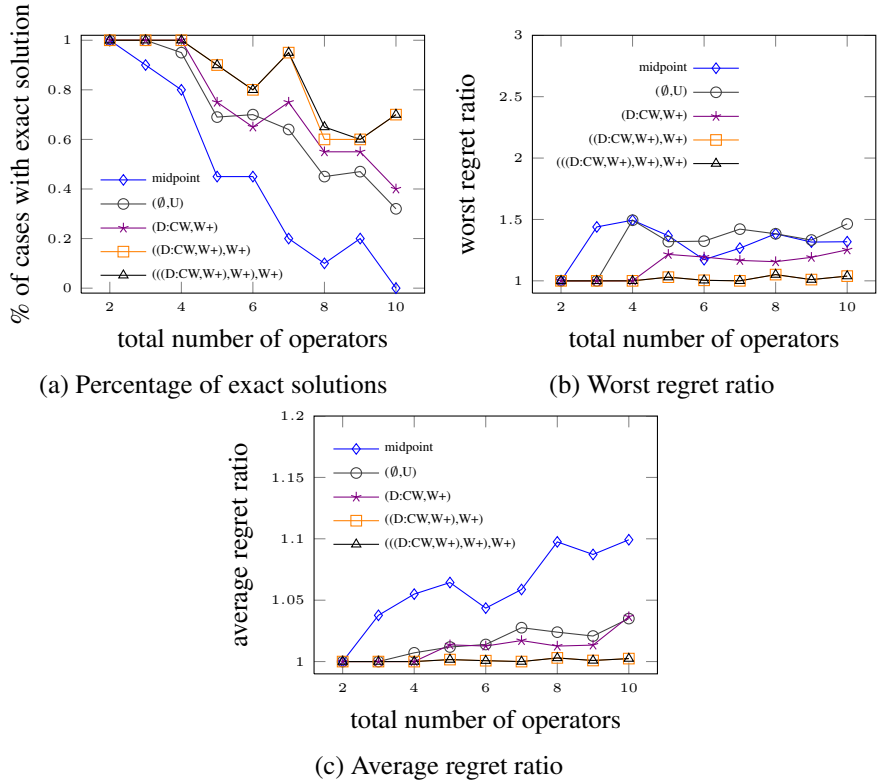

The results on the Enron data set (Figure~\ref{fig:enron}) showed similar 
trends\footnote{The run time was exactly the same, which is why we are omitting
	the diagram here.},
but were more impressive in every respect.
The two- and three-phase variants of the max-min heuristic
found the minmax optimal solution in 84\% of cases, had a
worst regret ratio of only 1.05, and an
average regret ratio of less than 1.001.
By contrast, the midpoint heuristic had a worst regret ratio of over 1.49, an average of 1.06,
and did not find a single minmax optimal solution with 10 operators

To highlight how bad a poor choice of selectivity can be,
we also tested using the minimum selectivity values of the intervals
(as would be done if estimates were based simply on
the selectivity of the keywords themselves).  This produced a worst case
regret ratio of {\em almost 30 for only 5 operators}.

\subsubsection{Star Schema Benchmark}

We optimised the generated SSB queries using minmax regret optimisation, a
mean-value-based approach, and also computed the optimal execution plan using
exact selectivities, which means that we are comparing actual query plan costs
rather than regret ratios.

Figure~\ref{fig:furtherresults}(a) 
shows the results for the average difference in 
costs between the query execution plans generated by 
different methods and the optimal plan (every data
point in the diagram averages the measurement obtained by 
running 100 different queries). 
We only include two variants of minmax regret optimisation,
(D:CW,W+) and the simple midpoint heuristic,
as for SSB no major differences were discernible between the different 
variants in terms of the quality of the query plans.
Surprisingly, the midpoint heuristic, although not very good at optimising the
regret ratio, seems to produce efficient query execution plans.
Considering the fact that
all queries had an average run time between 60 and 80 seconds, the numbers
shown in Figure~\ref{fig:furtherresults}(a) may not
seem like a big difference. However, this shows that minmax regret
optimisation delivers better plans than a mean-value-based approach.

More important is the robustness of the approaches, i.e., how good are they in
avoiding bad plans? Figure~\ref{fig:furtherresults}(b) 
shows the standard deviation of
the cost difference to the optimal plan, illustrating that the mean-value-based
approach is more erratic than minmax regret optimisation. The most extreme case
for all SSB queries was a mean-value-optimised plan more than
doubling the run time of the optimal plan (from 60s to 135s), while for minmax
regret optimisation the very worst plan added roughly 50\% more to the cost of 
the optimal plan (from 60s to 92s).


\begin{figure}
	\begin{tabular}{ccc}
		\begin{tikzpicture}
		\begin{axis}[%
		smallgraph,
		xmode=normal,
		ymode=normal,
		ymin=0.0,
		ymax=3.0,
		legend style={cells={anchor=east}, legend pos=north west, yshift=-1mm},
		legend plot pos=right,
		xlabel= total number of operators,
		ylabel= average diff in secs]
		\pgfplotstableread{./avg_ssb.txt}\avgssb
		\addplot[color=olive,mark=o] table[x=x_axis, y=mean] {\avgssb};
		\addlegendentry{{\smaller[2]mean value}}
		\addplot[color=violet,mark=star] table[x=x_axis, y=wplus] {\avgssb};
		\addlegendentry{{\smaller[2](D:CW,W+)}}
		\addplot[color=blue,mark=diamond] table[x=x_axis, y=midp] {\avgssb};
		\addlegendentry{{\smaller[2]midpoint}}
		\end{axis}
		\end{tikzpicture} 
		& \hspace*{-5.5mm} &
		\begin{tikzpicture}
		\begin{axis}[%
		smallgraph,
		xmode=normal,
		ymode=normal,
		ymin=0.0,
		ymax=9.0,
		legend style={cells={anchor=east}, legend pos=north west, yshift=-1mm},
		legend plot pos=right,
		xlabel= total number of operators,
		ylabel= std dev of diff in secs]
		\pgfplotstableread{./std_ssb.txt}\stdssb
		\addplot[color=olive,mark=o] table[x=x_axis, y=mean] {\stdssb};
		\addlegendentry{{\smaller[2]mean value}}
		\addplot[color=violet,mark=star] table[x=x_axis, y=wplus] {\stdssb};
		\addlegendentry{{\smaller[2](D:CW,W+)}}
		\addplot[color=blue,mark=diamond] table[x=x_axis, y=midp] {\stdssb};
		\addlegendentry{{\smaller[2]midpoint}}
		\end{axis}
		\end{tikzpicture} \\
		(a) Average difference in cost to opt &\hspace*{-5.5mm}&
		(b) Standard deviation of difference in cost to optimal \\
	\end{tabular}
	\caption{Results on difference in cost to optimal.}
	\vspace*{-.4cm}
	\label{fig:furtherresults}
\end{figure}
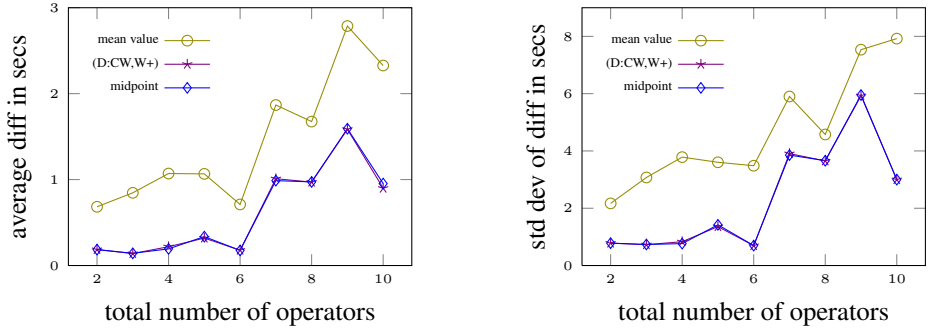

\section{Conclusion} 
\label{sec:conclusion}

We have investigated query optimisation under partial ignorance, in particular
ordering selection operators optimally if their selectivities are
defined by an interval rather than an exact value. 
The strategy we employed, minmax regret optimisation (MRO), is considered to
be a pessimistic approach compared to other techniques from decision
theory. In our opinion this makes it well-suited to query optimisation in
database systems, which should be about avoiding bad plans rather than finding
the best one. There is one major drawback, though: selection ordering becomes
NP-hard when applying MRO to it. However, we have shown that
special cases can be solved efficiently and that heuristics can quickly find
good solutions.

For future work we plan to extend our approach to costs described by intervals
and relative regret, i.e., considering the ratio of the cost of a plan to the
optimal plan for a scenario rather than the difference. Also interesting are 
other operators, such as
joins, whose ordering is heavily influenced by selectivities as well and suffers
from similar issues: it is hard to obtain exact values. Further topics we would
like to tackle are finding approximation algorithms with proven bounds and
modelling correlation of query predicates. Nevertheless, we think this is an
important first step in discovering new approaches for making query optimisers
more robust and one of our medium term goals is to build a general framework
for query optimisation under partial ignorance.

\balance
\bibliographystyle{plain}
\bibliography{./paperShort}


\end{document}